\documentclass[letterpaper,11pt]{article}
\usepackage[utf8]{inputenc}

\pdfoutput=1

\usepackage{amsmath, amsthm, amssymb}
\usepackage[linesnumbered,vlined]{algorithm2e}
\usepackage[table,xcdraw]{xcolor}
\usepackage{mathtools}
\usepackage[numbers]{natbib}
\usepackage{color-edits} % author comments with colors
\usepackage[most]{tcolorbox}
\usepackage{xfrac}
\usepackage{hyperref}
\usepackage{multirow}
\usepackage{caption}
\usepackage{bm}
\usepackage{newfloat}
\usepackage{xpatch}
\usepackage{soul}
\usepackage{bbm}
\usepackage{algorithmic}
\usepackage{thm-restate}
\usepackage[inline]{enumitem}

\makeatletter
\newcommand{\thickhline}{%
    \noalign {\ifnum 0=`}\fi \hrule height 1.5pt
    \futurelet \reserved@a \@xhline
}
\makeatother

\date{}

\usepackage[margin=1in]{geometry}

\allowdisplaybreaks

\definecolor{mygreen}{RGB}{20,120,60}

\title{Beating $(1-1/e)$-Approximation for Weighted Stochastic Matching
}

\author{ Mahsa Derakhshan\thanks{Northeastern University. Email: \texttt{m.derakhshan@neu.edu}} \and Alireza Farhadi\thanks{Carnegie Mellon University. Email: \texttt{farhadi94@gmail.com}}}

\newcommand{\E}[0]{\ensuremath{\mathbb{E}}}

\newcommand{\opt}[0]{\ensuremath{\textsc{OPT}}}
\newcommand{\ubar}[0]{\ensuremath{\sigma}}
\newcommand{\gp}[0]{\ensuremath{G_{\textbf{p}}}}

\newcommand{\sol}[0]{\ensuremath{\mathcal{M}}}
\newcommand{\W}[1]{\ensuremath{\mathcal{W}(#1)}}

\newcommand{\inm}[1]{\ensuremath{Z}}

\newcommand{\match}{\mathcal{M}}
\newcommand{\wdotx}{\omega}
\newcommand{\optlp}{\opt_{\text{LP}}}
\newcommand{\apximprove}{0.0014}

\newcommand{\curly}{\mathrel{\leadsto}}
\newcommand{\PAM}{Alg1}
\SetKwFunction{GKS}{GKS}
\SetKwFunction{PAM}{BaseMatching}
\SetKwFunction{FNM}{APXMatching}
\SetKwFunction{DIST}{DrawPermutation}

\newcommand{\emax}{\Delta}
\newcommand{\gm}{\varphi}

\newcommand{\odv}{\mathrm{d}}

\SetKwRepeat{Do}{do}{while}

\renewcommand{\b}[1]{\ensuremath{\bm{\mathrm{#1}}}}

\renewcommand{\epsilon}[0]{\ensuremath{\varepsilon}}

\let\originalleft\left
\let\originalright\right
\renewcommand{\left}{\mathopen{}\mathclose\bgroup\originalleft}
\renewcommand{\right}{\aftergroup\egroup\originalright}

\addauthor{sb}{blue}    % sb for Soheil
\addauthor{md}{red}    % md for Mahsa

\newtheorem{lemma}{Lemma}[section]

\newtheorem{definition}[lemma]{Definition}
\newtheorem{claim}[lemma]{Claim}
\newtheorem{fact}[lemma]{Fact}
\newtheorem{observation}[lemma]{Observation}

\makeatletter
\def\thm@space@setup{%
  \thm@preskip= 0.2cm
  \thm@postskip=\thm@preskip % or whatever, if you don't want them to be equal
}
\makeatother

\definecolor{mygreen}{RGB}{20,155,20}
\definecolor{myred}{RGB}{195,20,20}
\definecolor{linkcolor}{RGB}{0,0,230}
\definecolor{mylightgray}{RGB}{230,230,230}
\definecolor{verylightgray}{RGB}{240,240,240}
\definecolor{commentcolor}{RGB}{120,120,120}

\SetCommentSty{mycommfont}

\hypersetup{
     colorlinks=true,
     citecolor= mygreen,
     linkcolor= myred,
     urlcolor= mygreen
}

\newcommand{\etal}[0]{\textit{et al.}}

\newcounter{myalgctr}
\newenvironment{tbox}{
\par\addvspace{0.2cm}
\begin{tcolorbox}[width=\textwidth,
                  enhanced,
                  boxsep=2pt,
                  left=1pt,
                  right=1pt,
                  top=4pt,
                  boxrule=1pt,
                  arc=0pt,
                  colback=white,
                  colframe=black,
                  unbreakable
                  ]%%
}{
\end{tcolorbox}
}

\newenvironment{tboxh}{
\par\addvspace{0.2cm}
\begin{tcolorbox}[width=\textwidth,
                  enhanced,
                  boxsep=2pt,
                  left=1pt,
                  right=1pt,
                  top=4pt,
                  boxrule=1pt,
                  arc=0pt,
                  colback=white,
                  colframe=black,
                  unbreakable,
                  float=t
                  ]%%
}{
\end{tcolorbox}
}

\newenvironment{graytbox}{
\par\addvspace{0.1cm}
\begin{tcolorbox}[width=\textwidth,
                  enhanced,
                  frame hidden,
                  boxsep=5pt,
                  left=1pt,
                  right=1pt,
                  top=2pt,
                  bottom=2pt,
                  boxrule=1pt,
                  arc=0pt,
                  colback=mylightgray,
                  colframe=black,
                  breakable
                  ]%%
}{
\end{tcolorbox}
}

\newcommand{\tboxhrule}[0]{\vspace{0.1cm} \hrule \vspace{0.2cm}}

\newenvironment{titledtbox}[1]{\begin{tbox}#1 \tboxhrule}{\end{tbox}}
\newenvironment{titledtboxh}[1]{\begin{tboxh}#1 \tboxhrule}{\end{tboxh}}

\newenvironment{tboxalg2e}[1]{
\refstepcounter{myalgctr}
	\begin{titledtbox}{\textbf{Algorithm \themyalgctr.} #1}
	\vspace{-0.2cm}
}
{
	\vspace{-0.3cm}
	\end{titledtbox}
}

\begin{document}
\maketitle

\thispagestyle{empty}

In the {\em stochastic weighted matching} problem, the goal is to find a large-weight matching of a graph when we are uncertain about the existence of its edges. In particular, each edge $e$ has a known weight $w_e$ but is realized independently with some probability $p_e$. The algorithm may {\em query} an edge to see whether it is realized. We consider the well-studied {\em query commit} version of the problem, in which any queried edge that happens to be realized must be included in the solution.

\smallskip
Gamlath, Kale, and Svensson [SODA'19] showed that when the input graph is bipartite, the problem admits a  $(1-1/e)$-approximation. In this paper, we give an algorithm that for an absolute constant $\delta > \apximprove$ obtains a $(1-1/e+\delta)$-approximation, therefore breaking this prevalent bound.

{

\hypersetup{
     linkcolor= black
}
\newpage
\thispagestyle{empty}
\tableofcontents{}
\clearpage
}

\setcounter{page}{1}

\section{Introduction}

%\mdcomment{realized or exist? be consistent}

Matching under uncertainty is a well-studied topic in theoretical computer science, primarily due to applications in matching markets such as kidney exchange, labor markets, online advertising, etc. (See~\cite{DBLP:conf/stoc/TangW020, DBLP:conf/focs/FeldmanMMM09, DBLP:conf/focs/BehnezhadD20, DBLP:conf/soda/Singla18} and the references within.)  In this paper, we study the max-weight \emph{stochastic matching} problem on bipartite graphs in the \emph{query commit} model. We are given a bipartite graph $G = (A\cup B, E)$ along with a weight $w_e$ and a parameter $p_e \in (0,1]$ for any edge $e\in E$. Each edge of this graph is realized (or exists) independently, with a known probability $p_e$ forming the realized subgraph  $G_p$. We are unaware of the realization of the edges but can access them via queries. Our goal is to find a large-weight matching of $G_p$ in the query commit model in which any edge joins the matching if and only if it is queried and realized. In other words, to add an edge to the matching, we have to ensure that it exists; thus, we have to query it. Moreover, if a queried edge is realized, we have to add it to our matching (i.e., commit to it). This also implies that if an edge has at least one matched end-point it cannot be queried. 

%In this paper, we focus on designing an approximation algorithm for the above problem, where the approximation ratio is defined with respect to the max-weight matching found by an algorithm that knows the realization of all the edges. For this problem, a $0.5$-approximation can be achieved via the greedy matching algorithm by simply iterating over the edges in the decreasing order of their weights and querying any edge with unmatched end-points. The first non-trivial algorithm for this problem designed by Gamlath, Kale, and Svensson~\cite{DBLP:conf/soda/GamlathKS19} achieves a  $(1-1/e)$ approximation ratio, a very familiar bound that holds for many matching algorithms in various settings (e.g.,~\cite{DBLP:conf/soda/GoelM08,  BDM, DBLP:conf/stoc/TangW020, DBLP:journals/corr/abs-2002-06037, DBLP:conf/soda/GoelM08}). It also remains the best-known approximation ratio for this problem until our work.  In this paper, we provide a novel algorithm, breaking this $(1-1/e)$ barrier. \mdcomment{mention that their analysis is tight.}

In this paper, we focus on designing an approximation algorithm for the above problem, where the approximation ratio is defined with respect to the max-weight matching found by an algorithm that knows the realization of all the edges. For this problem, a $0.5$-approximation can be achieved via the greedy matching algorithm by simply iterating over the edges in the decreasing order of their weights and querying any edge with unmatched end-points. The first non-trivial algorithm for this problem designed by Gamlath, Kale, and Svensson~\cite{DBLP:conf/soda/GamlathKS19} achieves a  $(1-1/e)$ approximation ratio, a very familiar bound that holds for many matching algorithms in various settings (e.g.,~\cite{DBLP:conf/soda/GoelM08,  BDM, DBLP:conf/stoc/TangW020, DBLP:journals/corr/abs-2002-06037, DBLP:conf/soda/GoelM08}). It is also known that $(1-1/e)$ is the best approximation ratio achievable by Gamlath et al.'s algorithm~\cite{DBLP:journals/corr/abs-2002-06037}, which remains to be the state of the art for this problem.  In this paper, we provide a novel algorithm, breaking this $(1-1/e)$ barrier.

\vspace{3mm}
\begin{graytbox}
\noindent \textbf{Main Theorem.} (See Theorem~\ref{thm:main1}) 
There exists a polynomial-time algorithm that, given any stochastic weighted bipartite graph, finds a matching with an approximation ratio of at least $1-1/e+\apximprove$ in the query commit model.
\end{graytbox}
\vspace{3mm}

After Gamlath et al.'s $(1-1/e)$-approximation algorithm, Tang, Wu, and Zhang~\cite{DBLP:conf/stoc/TangW020, DBLP:journals/corr/abs-2002-06037}\footnote{ \cite{DBLP:journals/corr/abs-2002-06037} is an extension of~\cite{DBLP:conf/stoc/TangW020} published separately by the same authors. 
%They mention in~\cite{DBLP:conf/stoc/TangW020} that this decision was to avoid distraction.
} showed that this approximation ratio can also be extended to the more general \emph{oblivious matching} problem using a random greedy algorithm.
Unlike our problem, in which edges exist independently with known probabilities, in the oblivious matching problem, an oblivious adversary decides which edges exist\footnote{The adversary is called oblivious since they determine which edges exist, beforehand, oblivious to the randomization of the algorithm.}. 
A natural question here is whether there is a separation between the two models when it comes to weighted bipartite graphs. In this work, we take a step toward answering this question as we break the $(1-1/e)$ barrier for the stochastic matching problem while the best-known approximation ratio for the oblivious matching problem remains to be $(1-1/e)$.

% Tang, Wu and Zhang provide a different algorithm for the \emph{oblivious matching} problem, achieving a $(1-1/e)$-approximation ratio for weighted bipartite graphs. Oblivious matching can be perceived as a generalization of the problem studied in this paper. An open question here is whether having the distribution over edge realizations actually helps the approximation ratio. In this paper, we take 

\subsection{Our Techniques}\label{section:tecnique}

In this section, we will give a brief overview of our ideas and techniques which allow us to design an algorithm with an approximation factor larger than $(1-1/e)$. Our algorithm is LP-based. We start with an LP   which we borrow from Gamlath et al.~\cite{DBLP:conf/soda/GamlathKS19}. This LP gives us an upper-bound for the optimal solution of the problem. We then design a rounding procedure that given the optimal solution of this LP, decides which edges to query and in which order. In this section, we first describe the LP. We then briefly explain how Gamlath et al. round the solution of this LP in order to find a $(1-1/e)$-approximate matching. Finally, we discuss our rounding procedure which results in beating the $(1-1/e)$-approximation.

\paragraph{The LP.} We will use an LP designed by Gamlath et al. The LP has a variable $x_e$ for any edge, and its solution forms a fractional matching. Besides the natural matching constraints, the LP also has a special constraint for any subset of edges sharing an end-point to ensure that the probability of at least one of them being realized is not smaller than the total fraction of matching on them (Constraint~\eqref{cons:437489} below). The LP clearly has exponentially many constraints, but Gamlath et al. show that it is polynomially solvable using the ellipsoid method. Below we provide the formulation of this LP in which $E_u$ denotes the subset of edges that have $u$ as one of their end-points.

\vspace{2mm}
\newlength{\LPlhbox}
\settowidth{\LPlhbox}{(LP-Match)}%
\noindent%
\begin{minipage}{\linewidth-2cm}
 \begin{align}
\nonumber \max_{\mathbf{x}} \qquad &\b{x} \cdot \b{w}	\,& \\
\textrm{s.t.} \qquad & \sum_{e \in F} x_e \le \Pr[\text{an edge in $F$ exists}]& \forall u \in A \cup B,  \forall F \subseteq E_u  \label{cons:437489}\\
&x_e \ge 0& \forall e \in E 
\end{align}
\end{minipage}
\vspace{10mm}

Now, let us briefly explain the $(1-1/e)$-approximation algorithm  of Gamlath et al.~\cite{DBLP:conf/soda/GamlathKS19}.  In the rest of the paper, we will refer to this algorithm as \GKS (The authors' initials). 
\paragraph{Gamlath et al.'s rounding procedure.} Based on the solution of the LP, they draw a random ordering over the edges satisfying some special properties. They then iterate over the edges in this order and decide whether to query them and subsequently add them to the matching or lose them forever.  The query decisions are made dynamically and they consist of each vertex in part $B$ of the graph running a prophet-secretary based algorithm~\cite{DBLP:conf/soda/EhsaniHKS18} to decide which one of its edges should be queried. Due to making irrevocable query decisions, the algorithm may choose not to query some of the edges while both their end-points are unmatched at the end of the algorithm.
 
 To prove that \GKS finds a a $(1-1/e)$-approximate matching, they show that the expected weight of the matching edge connected to any vertex in set $u\in B$ is at least $(1-1/e)$ fraction of the total contribution of its edges to the optimal solution of the LP. That is  $$\sum_{e\in E_u}\Pr[\GKS  \text{ adds } e \text{ to the matching}]\cdot w_e \geq (1-1/e) \sum_{e\in E_u} x_ew_e.$$  
 Later, Fu et al.~\cite{DBLP:conf/icalp/FuTWWZ21} observe that replacing the prophet-secretary based approch with an OCRS algorithm makes it possible for each edge to join the matching with probability at least $(1-1/e)x_e$.

\paragraph{Our rounding procedure.} Our rounding procedure consists of two steps. We first design a rounding procedure, which, given the optimal solution of the LP, returns a $(1-1/e)$-approximate matching and prove some new desirable properties for it. We call this procedure \PAM. We then use this as a subroutine in our final algorithm which beats $(1-1/e)$-approximation. Below are the four important properties. 
\begin{enumerate}[label=(P\arabic*)] 
\item \label{prop:two} For any vertex $u\in B$, let $\ubar_u$ denote the probability with which this vertex is matched in the optimal solution of the LP. That is $\ubar_u=\sum_{e\ni u} x_e$. One of the desirable properties of \PAM is that if $\ubar_u$ is smaller than one by a constant, then  its edges will have a higher chance of joining the matching in comparison to the case of $\ubar_u=1$ (by a constant fraction). As a result of this property, if $\ubar_u$ is small enough for all the vertices in $u\in B$, \PAM itself achieves an approximation ratio larger than $(1-1/e)$.	
	\item \label{prop:three} After running \PAM, any vertex of the graph remains unmatched with a constant probability. Similarly, any edge does not join the matching with a probability of at least $\Omega(x_e)$. We will later explain why this property is helpful.
 	\item  \label{prop:four}  For any edge $e=(v,u)$, we want to ensure that conditioned on that $e$ is not queried by \PAM, there is a constant probability of both its end-points being unmatched. Property~\ref{prop:three} already implies that $u$ and $v$ remain unmatched with a large probability.  This property, however, is about the correlation between these events and is the most challenging  one to prove.
 We will use this property to show that for some instances of the problem running $\PAM$ for a second time on the remaining edges results in an improved approximation ratio.
\end{enumerate}

In order to achieve these properties, \PAM takes the optimal solution of the LP, makes several changes to it and then feeds it to a modified version of \GKS (which assumes all the edges have the same weight). One of these changes relies on a transformation function. This function takes $\b{x}$, the optimal solution of the LP and outputs $\tilde{\b{x}}$ which is also a solution of the LP. For any edge $e$, we set $\tilde{x}_e=g(x_e, \ubar)$ where $\ubar$ is a (carefully picked) constant no smaller than $\max_{u\in B} \sum_{e\ni u} x_e$, and

$$g(x, \ubar) = \frac{(e^\ubar-1)(\ubar  -x) x}{\ubar (e^\ubar- e^x)}.$$
It is not hard to see that for any $x_e$, this transformation results in $x_e \geq \tilde{x}_e$, therefore $\tilde{\b{x}}$ will still be a valid solution of the LP.    See Appendix~\ref{sec:transform} for more details about this function and the ideas behind choosing it. 

We will use properties \ref{prop:three} and \ref{prop:four} to prove that after running \PAM, any edge $e$ with a small enough $\tilde{x}_e/p_e$ (smaller than a constant) remains \emph{available} with a constant probability. We say an edge is available if it is not queried and its end-points are unmatched. Roughly speaking, for any edge $e$ satisfying this constraint, property \ref{prop:three} of \PAM implies that the edge is not queried with a constant probability and any of its end-points are also unmatched with a constant probability. The challenging part, however, is to prove that despite potential correlations, the probability of all three events happening together is still a constant. Therefore, the edge remains available with a constant probability. We use this feature of \PAM in a crucial way in order to break the $(1-1/e)$ approximation ratio. The reason being that if edges with a small enough $\tilde{x}_e/p_e$ constitute a constant fraction of the optimal solution of the LP, then we can run \PAM for a second time on the remaining available edges and use its output to enhance the matching.

\paragraph{Beating $(1-1/e)$-approximation}
Now, let us briefly explain how we use the aforementioned properties  of \PAM to design another algorithm with an improved approximation ratio. We claim that for any instance of the problem, rounding the optimal solution of the LP in at least one of the two following ways results in an improved approximation ratio. \begin{enumerate*}[label=(\roman*)]
 \item simply run \PAM once on the whole graph and once on the remaining available edges.
 \item   Prune the graph by deleting any edge with $\tilde{x}_e/p_e<\tau$   for a carefully picked threshold $\tau$ (which is a constant smaller than one), and then run \PAM on the remaining subgraph.
 \end{enumerate*}
 
As we already discussed, we prove that edges with  $\tilde{x}_e/p_e<\tau$ will remain available with a constant probability after we run \PAM. Therefore, if at least a constant fraction of the optimal solution of the LP comes from these edges, then the first rounding achieves an improved approximation ratio. On the flip side, if these edges do not have a large contribution, by deleting them we are left with a solution of the LP which is close to optimal. We now claim that simply running \PAM on the remaining subgraph gives us a matching with an improved approximation ratio. Roughly speaking, we are able to take advantage of constraint \eqref{cons:437489} of the LP, in order to show the following for any vertex $u$:  The sum of ${x}_e$ of edges connected to $u$ that satisfy  $\tilde{x}_e/p_e\geq \tau$ are upper-bounded by a small enough constant. This implies that in the remaining subgraph, $\sigma_u$ of all the vertices in bounded away from one by a constant. Invoking property~\ref{prop:two} of \PAM, implies that simply running it on this remaining subgraph results in an improved approximation ratio. Using  constraint \eqref{cons:437489} of the LP in this way is a technically challenging part of the paper which we discuss in Section~{\ref{section:approx-ratio}}.

\subsection{Further Related Work}
The stochastic matching problem in the query commit model was first considered   by  Chen et al. \cite{DBLP:conf/icalp/ChenIKMR09}, who study the problem under the additional constraint that for any vertex a limited number of its edges can be queried. In this setting, they provide a $0.25$-approximation algorithm for bipartite unweighted graphs. After a series of works~\cite{DBLP:journals/algorithmica/BansalGLMNR12,DBLP:journals/algorithmica/BavejaCNSX18, DBLP:journals/corr/abs-1110-0990}, this approximation ratio was improved to $0.5$ for unweighted graphs~\cite{DBLP:journals/ipl/Adamczyk11}, $0.39$ for bipartite weighted graphs~\cite{DBLP:journals/corr/abs-2010-08142}, and $0.269$ for general wighted graphs~\cite{DBLP:journals/corr/abs-2010-08142}. (In all  these works, the approximation ratio is defined with respect to the online benchmark.)  Costello et al.~\cite{DBLP:conf/icalp/CostelloTT12} are the first to consider the query commit model without the above-mentioned extra constraint. In this setting, they provide a $0.573$-approximation algorithm for general unweighted graphs. Subsequently, the problem was studied for general wighted graphs with the best known approximation ratio being $0.567$ due to~\cite{DBLP:conf/icalp/FuTWWZ21}.

A problem, closely connected to the one studied in this paper, is the \emph{online bipartite matching} problem first introduced by Karp, Vazirani, and Vazirani~\cite{DBLP:conf/stoc/KarpVV90}. Here, vertices in one part of the graph are present from the beginning while vertices on the other part arrive online. Whenever an online vertex arrives, its edges are revealed, and the algorithm has to decide whether to match the vertex with an unmatched neighbor, irrevocably, or lose it forever. For this problem, Karp et al. present their celebrated ranking algorithm which achieves a tight $(1-1/e)$ approximation ratio in the standard adversarial model. Later, Mahdian and Yan~\cite{DBLP:conf/stoc/MahdianY11} show that if the online vertices arrive in a random order, the ranking algorithm achieves an approximation ratio of at least $0.69$. On the other hand,  Karande et al.~\cite{DBLP:conf/stoc/KarandeMT11} provide an upper-bound of $0.727$ for the ranking algorithm in the random arrivals model. Since ranking is a random greedy algorithm, it can also be used for the problem studied in this paper achieving an approximation ratio of  $0.69$ for unweighted graphs.  This is indeed the best-known approximation ratio for the unweighted version of the bipartite stochastic matching problem in the query commit model (the unweighted version of our problem).

Our paper also relates to a line of work initiated by Blum et al.~\cite{DBLP:journals/ior/BlumDHPSS20} which concerns the query complexity of the stochastic matching problem. In this model, the goal is to find a large matching of the graph while having only a constant number of queries per vertex. Here, the matching does not have to commit to the queried edges that are realized.  After a long line of work on this problem, (See~\cite{DBLP:conf/soda/YamaguchiM18, DBLP:conf/sigecom/BehnezhadR18, DBLP:conf/stoc/BehnezhadDH20} and the references within.)  Behnezhad et al.~\cite{DBLP:conf/focs/BehnezhadD20}   prove that it is possible to  achieve a $(1-\epsilon)$ approximation ratio for general weighted graphs.

Another model which is closely related to the one considered in this paper, is the price of information model (POI). It was first introduced by Singla~\cite{DBLP:conf/soda/Singla18} and subsequently studied in~\cite{DBLP:conf/icalp/FuTWWZ21, DBLP:conf/soda/GamlathKS19}.  In POI, (similar to  the query commit model), to know the edge weights one needs to query them, however, instead of having to commit to the realized edges, they need to pay a cost for each query. Moreover, edge weights come from discrete distributions. The best known approximation ratio for matching in the POI model is a $(1-1/e)$ by Gamlath et al. ~\cite{DBLP:conf/soda/GamlathKS19}. Due to their similar nature,  the POI model and the query commit model are often studied together. In Appendix~\ref{sec:poi} we have a discussion about extending our ideas to the POI model in order to improve over the  $(1-1/e)$ approximation ratio.

\section{Preliminaries}
We consider the stochastic matching problem on weighted bipartite graphs in the query-commit model. Our input is a  bipartite graph $G=(A \cup B,E) $, a vector of weights $\b{x} = (x_e)_{e\in E}$, and vector of probabilities  $\b{p} = (p_e)_{e\in E}$. Each edge exists or is realized with probability $p_e$ forming subgraph $\gp$. We refer to $\gp$ as the realized subgraph.  We do not know which edges are realized but can access them via queries.  We can query any edges $e\in E$ to see whether it exists or not, and if it does, we must add it to the matching irrevocably. In other words, our matching must commit to the edge.  Our goal is to design a polynomial-time algorithm which finds a matching $\sol$ of $\gp$ with a large approximation ratio. The approximation ratio of the algorithm is defined as
$$\frac{\E[\W{\sol}]}{\E[\W{\opt}]},$$
where $\opt$ is the maximum weight matching of $\gp$ (i.e., the optimal solution) and  $\W{.}$ is a function outputting the weight of a given matching. The expectation here is taken over both the randomization in the algorithm and realization of edges.

\paragraph{General notations.} For any $v\in A\cup B$ we will use $E_v$ to denote the subset of edges that have $v$ as one of their end-points. Unless otherwise stated, we will refer to vertices of $B$ using the letter $u$ and vertices of $A$ using the letter $v$. Moreover, we will write our edges as ordered pairs, meaning that for any edge $(v,u)\in  E$ we always have $v \in A$ and $u \in B$.

 \paragraph{Organization of the paper.} The rest of the paper is organized as follows. In Section~\ref{section:LP} we explain the LP which we borrow from Gamlath et al. This LP gives us an upper-bound on the optimal solution of the LP. In Section~\ref{section:rounding}, we first design \PAM, a rounding procedure which given the optimal solution of the LP outputs a $(1-1/e)$-approximate matching. Later in the section, we design \FNM which uses \PAM to find a matching with an approximation ratio of $(1-1/e+\apximprove)$. In Section~\ref{section:properties}, we prove some desirable properties of our \PAM, and finally, in Section~\ref{section:approx-ratio} we use these properties to analyze the approximation ratio of our algorithm.

%Let use $v \curly u$ to denote the event that $v$ proposes to the $u$ during the first round of the algorithm, i.e., $v$ queries the edges $(v,u)$ in the first round and this edge gets realized. Similarly, we use $v \not \curly u$ to denote the event that $v$ is not proposed to $u$ in the first round, i.e., the edge $(v,u)$ is either not queried or if it is queried, it should not be realized.

\section{The Linear Program}\label{section:LP}
In this section, we describe our algorithm which consists of two parts. The first part is an LP which gives an upper-bound on the optimal solution of the problem. We then design a rounding procedure that given any solution of this LP, decides which edges to query and in which order. The LP which we borrow from~\cite{DBLP:conf/soda/GamlathKS19} is as follows.

\setcounter{equation}{0}
\vspace{3mm}
%\newlength{\LPlhbox}
\settowidth{\LPlhbox}{(LP-Match)}%
\noindent%
\parbox{\LPlhbox}{\begin{align}
               \tag{LP-Match}\label{LP-Match}
         \end{align}}%
\hspace*{\fill}%
\begin{minipage}{\linewidth-2cm}
 \begin{align}
\nonumber \max_{\mathbf{x}} \qquad &\b{x} \cdot \b{w}	\,& \\
\textrm{s.t.} \qquad & \sum_{e \in F} x_e \le \Pr[\text{an edge in $F$ exists}]& \forall u \in A \cup B,  \forall F \subseteq E_u \label{cons:krfjerkf} \\
&x_e \ge 0& \forall e \in E 
\end{align}
\end{minipage}
\vspace{5mm}

In the LP above $x_e$ can be perceived as the probability of $e$ being in the optimal matching. For a vertex $u$ and any subset $F \subseteq E_u$,  $\sum_{e \in F} x_e$ is the probability that at least one of the edges in $F$ is in the optimal solution. Hence,  $\sum_{e \in F} x_e$ should be bounded by the probability that at least one of the edges in $F$ is exists. Since each edge  $e\in F$ is realized independently with probability $p_e$, we have $\Pr[\text{an edge in $F$ exists}] = 1- \prod_{e \in F} (1-p_e)$. 

\begin{lemma}[\cite{DBLP:conf/soda/GamlathKS19}]
	The optimal solution of \ref{LP-Match} is an upper-bound for \opt. Moreover, this solution can be found in polynomial time.
\end{lemma}

A nice property of this LP which is crucial for the $(1-1/e)$-approximation algorithm of Galamth~\etal{} is as follows. For any vertex $v$, it is possible to randomly permute its edges such that for any edge $e$ the probability of it being the first realized edge in the permutation is no smaller than $x_e$. To make this an exact equality, sometimes this permutations needs to be over a subset of $E_v$ instead of the whole set.  Below, we formally state this property.

\begin{definition}[proportional distributions] \label{def:irifrjf}
Let {\b x} be a solution of \ref{LP-Match} (not necessarily the optimal solution). We say a distribution $\mathcal{D}^{\b x}_v$ over permutations of subsets of $E_v$  is proportional to $\b{x}$ if  for any edge $e\in N_v$ the probability that $e$ is outputted by the following algorithm (Algorithm~\ref{alg:def}) is exactly equal to $x_e$.
\end{definition}

\begin{tboxalg2e}{}
\begin{algorithm}[H] \label{alg:def}
% \caption{Vertex and pair oracles.}
% \label{alg:edgeoracle}
% \caption{``vertex oracle''}	
    \DontPrintSemicolon
        Let $\pi_v\sim \mathcal{D}^{\b x}_v$.\;
        \ForEach{$e$ permuted in the order of  $\pi_v$}{
        Query edge e.\;
        \If{$e$ is realized}
        {Output $e$ and end the algorithm.}
        }
\end{algorithm}
\end{tboxalg2e}

\begin{lemma}[\cite{DBLP:conf/soda/GamlathKS19}]\label{lemma:dist}
	Let {\b x} be a solution of \ref{LP-Match} (not necessarily the optimal solution). For any vertex $v$, there exists a a distribution $\mathcal{D}^{\b x}_v$ over permutations of subsets of $E_v$ that is proportional to $\b{x}$. Moreover, it is possible to draw a permutation from this distribution in polynomial time.
	
	\end{lemma}

\section{The Rounding Procedure}\label{section:rounding}

Our rounding procedure consists of two parts. First is an algorithm which outputs a $(1-1/e)$ approximate matching and has some new desired properties. We will refer to this as \PAM. The second part is using \PAM in order to design  another rounding procedure with approximation ratio larger than $(1-1/e)$.

Before presenting \PAM, we first explain a a simplified version of it which outputs a $(1-1/e)$ approximate matching for unweighted graphs. (This is also a simplified version of \GKS.) Let us start by providing a definition which we will use in the rest of the paper.

\begin{definition}[Examining an edge $e$]
We examine an edge to determine whether it is realized or not. If both its end-points are unmatched, examining $e$ is equivalent to querying it. It is always guaranteed that in this case, $e$ joins the matching if it is realized. 
  
At some points in our algorithm, we also need to use the information about realization of an edge who has at least one matched end-point. In this scenario, we cannot query the edge because it cannot join the matching. Instead, we just toss a coin which with probability $p_e$ decides that the edge is realized. In such cases, it is guaranteed that the edge will not join the matching. 
\end{definition}

\SetKwFunction{SPM}{SimpleMatching}

\paragraph{The simplified algorithm.} This algorithm is stated formally as \SPM  below. In this algorithm, we first draw a permutation $\pi$ over vertices in $A$ uniformly at random. Going over $A$ in this order, for any vertex $v\in A$, we draw a permutation $\pi_v$ over its edges from $\mathcal{D}^{\b x}_v$, a distribution proportional to $\b{x}$.  (See Definition~\ref{def:irifrjf}.) We then examine edges of $A$ in this order and stop when one of them is realized. (In none of them is realized we do nothing.) Let us denote this edge by $e=(v,u)$. Vertex $v$ then sends a proposal to $u$. If $u$ is unmatched, it accepts the proposal and $e$ joins the matching. If $u$ is already matched, it rejects the proposal and $v$ remains unmatched forever. The algorithm then proceeds to repeat this for the next vertex in permutation $\pi$.

\begin{tboxalg2e}{A $(1-1/e)$-approximation algorithm for unweighted graphs.}
\begin{algorithm}[H]\label{alg:simple}
% \caption{Vertex and pair oracles.}
% \label{alg:edgeoracle}
% \caption{``vertex oracle''}
    \SetKwProg{Fn}{Function}{:}{}	
    \DontPrintSemicolon
    \Fn{\SPM{$G=(A, B, E), \b x$}}{
    Let $\match=\emptyset$ be a matching of $E$.\;
    Let $\pi$ be a permutation over vertices in $A$ chosen uniformly at random. \;
\ForEach {$v$ in the order of $\pi$}{
        Let $\mathcal{D}^{\b x}_v$ be a distribution over permutations of subsets of $E_v$ proportional to \b{x}.\;
        Let $\pi_v\sim \mathcal{D}^{\b x}_v$.\;
        \ForEach{$e=(v,u)$ permuted in the order of  $\pi_v$}{
        Examine edge $e$.\\
        \If{$e$ is realized}{ Vertex $v$ sends a proposal to $u$.\\
        If $u$ is unmatched, it accepts the proposal and $e$ joins the matching.\\
        Terminate the loop.}
        }
        }
        \Return $\match$.
        }
\end{algorithm}
\end{tboxalg2e}
\vspace{4 mm}

We claim that \SPM{$G, \b x$} matches any vertex $u\in B$  with probability at least $$(1-1/e)\sum_{e\in E_u} x_e.$$ If we consider unweighted graphs, this implies a lower-bound of $(1-1/e)$ for the approximation ratio. Observe that this algorithm matches vertex $u$ it receives at least one proposal from its neighbors. Since $\pi_v$ is drawn from $\mathcal{D}^{\b x}_v$, by Definition~\ref{def:irifrjf}, the probability of $u$ receiving a proposal from any neighbor $v$ is exactly $x_{(v,u)}$. Combining this with the fact that any vertex $v\in A$ sends its proposals independently of other vertices in $A$, we can write:

\begin{align*}\Pr[u \text{ is matched }] &= \Pr[u \text{ receives at least one proposal}]\\
& \geq  (1-1/e)\sum_{e\in E_u} x_e \hspace{45 mm} \text{ (due to Fact~\ref{fact1} stated below)}\\
\end{align*}

 \begin{fact}\label{fact1}
 Let $x_1, \dots, x_k$ be a set of independent Bernoulli random variables with  $\E[X] \ge 0$, where $X$ is their sum. We have $$\Pr\left[X > 0\right] \geq 1-e^{-\E[X]}\geq \E[X](1-1/e).$$
 \end{fact}

 As mentioned before, this is also a simplified version of \GKS. To make this algorithm work for weighted graphs, they allow each vertex $u\in B$ to decide which proposal it is going to accept instead of just accepting the first one. To make this decision, for each vertex $u$ they run an algorithm similar to the one designed for the prophet secretary problem by Ehsani et al.~\cite{DBLP:conf/soda/EhsaniHKS18} which results in 
 $$\sum_{e\in E_u}\Pr[e \text{ joins the matching}].w_e \geq  (1-1/e) \sum_{e\in E_u}x_ew_e,$$ for any vertex $u\in B,$
 and thus a $(1-1/e)$ approximation algorithm. Later, Fu et al.~\cite{DBLP:conf/icalp/FuTWWZ21} observe that using an OCRS based approach (from  Lee and Singla~\cite{DBLP:conf/esa/LeeS18}) allows to turn this to a per-edge approximation. That is, they show that if each vertex $u$ runs an OCRS algorithm to decide which proposal it accepts, it is possible ensure that each edge joins the matching with probability at least $(1-1/e)x_e$.
 Let us emphasize that this is a rephrased version of these algorithms for the sake of comparison with ours. Of course, to describe their algorithm using our language, we need to change what examining an edge means, but we will not go into more details here.

\subsection{Our $(1-1/e)$-approximation Algorithm} 
\label{subsec:base}
Let us start by explaining the desirable properties that we are trying to achieve by designing a new $(1-1/e)$-approximate matching below. We call this algorithm \PAM.  We first define $$\ubar=\max_{u\in B} \sum_{e\in E_u} x_e.$$ These properties as follows. (We have briefly discussed these properties in Section~\ref{section:tecnique}.)
 
 \begin{enumerate}[label=(P\arabic*)]
\item \label{item:twotwo} The probability of any edge joining the matching is at least $(1-1/e)x_e$ and is also a  decreasing function of $\ubar$. In other words, if  $\ubar$ is smaller than one by a constant, we want \PAM to match any edge with probability at least $(1-1/e+c)x_e$, where $c$ is a constant. This implies that in this scenario, \PAM itself achieves an approximation ratio larger than $(1-1/e)$.
 	\item \label{item:threethree} The algorithm ensures that any vertex of the graph remains unmatched with a constant probability. Similarly, any edge does not join the matching with probability at least $\Omega(x_e)$. \SPM  does not satisfy this property as it may match some of the edges and vertices with probability one. This property allows us to show that for some instances of the problem, running \PAM twice gives us an improved algorithm. That is, we once run \PAM on $G$ and once on the edges that are not queried and have unmatched end-points after the first run.
 	\item \label{item:fourfour} For any edge $e=(v,u)$, we want \PAM to ensure that conditioned on any edge $e$ not being examined, there is a constant probability of both its end-points being unmatched. This, again is helpful in arguing that for instances of the problem running \PAM twice gives us an improved algorithm. Property~\ref{item:threethree} already implies that all the vertices remain unmatched with a large probability. It can also be used to show that some particular edges remain unexamined with a constant probability. This property, however, is about the correlation between these  events. 
 \end{enumerate}

Algorithm \PAM is very similar to \SPM with three  main differences. The first one which allows us to satisfy property~\ref{item:twotwo} is due to a transformation function. Given the optimal solution of the LP, we transform it to a different solution which still satisfies all the LP constraints. We then give this as an input to a modified version of \SPM. 

To pick our transformation function we start by observing that if all the edges connected to a vertex $u\in B$ have the same $x_e$ and $x_e\rightarrow 0$, then in \SPM each edge joins the matching with a probability that satisfies this property. Since the algorithm draws a uniform permutation over vertices in $A$, all of $u$'s edges have the same probability of joining the matching and that is

\begin{align*}\frac{\Pr[u \text{ joins the matching}]}{|E_u|}=\frac{1-(1-x_e)^{|E_u|}}{|E_u|}=\frac{1-(1-x_e)^{x_u/x_e}}{x_u/x_e}= \frac{(1-e^{-x_u})}{x_u}\cdot x_e. \end{align*}

Note that here $\frac{(1-e^{-x_u})}{x_u}$ is a decreasing function of $x_u$ and for $x_u=1$ it is exactly equal to $(1-1/e)$. Therefore, it satisfies property~\ref{item:twotwo}. Roughly speaking, based on this observation, we aim to modify the solution of the LP in a way that each edge $e=(v,u)$ behaves as if instead of a single edge it was a collection of edges $e'_1\dots e'_k$ with $k \rightarrow \infty$ and $\sum_{i=1}^k x_{e'}=x_e$. Of course, we cannot simply replace $e$ with $k$ parallel edges between $v$ and $u$ and hope to achieve the desired bound since these edges all being representatives of a single edge cannot behave independently. However, we show that lowering $x_e$ to $\tilde{x}_e$ defined as follows gives us a similar effect.

Let us denote the modified solution by ${\tilde{\b x}}$. Assume that for every $u \in B$, we have $\sum_{e \ni u} x_e \le \ubar$. For any edge $e\in E$ we let  $\tilde{x}_e=g(x_e, \ubar)$ where,
$$g(x, \ubar) = \frac{(e^\ubar-1)(\ubar  -x) x}{\ubar (e^\ubar- e^x)}$$
 for any $x\in [0,\ubar)$, and $$g(\ubar, \ubar) = \lim_{x\rightarrow \ubar} g(x, \ubar) = 1-e^{-\ubar}.$$ We then set $\tilde{x}_e = g(x_e, \ubar)$. When value of $\ubar$ is clear from the context we will simply use $g(x_e)$ instead of $g(x_e, \ubar)$.  In Appendix~\ref{sec:transform} we provide more details about this function and the ideas behind choosing it.

 \begin{restatable}{claim}{xless}
 For any $x\in [0,\ubar]$ and $\ubar\in (0,1]$	we have $x\geq g(x,\ubar)$.
 \end{restatable}
 The proof of this claim is deferred to Appendix \ref{sec:proofs}.

 \begin{observation}
 	The modified solution ${\tilde{\b x}}$ defined above satisfies all the constraints of \ref{LP-Match} since for any $e\in E$, we have $\tilde{x}_e= g(x_e, \ubar) \le x_e.$
 \end{observation}

It is possible to show that this modification ensures that the simple algorithm matches each edge with probability $(1-1/e)x_e$. This is a corollary of Lemma~\ref{lem:prm1} which we prove later. Indeed, the lemma implies the following lower-bound for the probability of each edge $e$ being matched by \SPM{$G, {\tilde{\b x}}$}.
$$\frac{(1-e^{-\ubar})x_e}{\ubar},$$ where $\ubar=\max_{u\in B}x_u$. Note that this is exactly what we were hoping to get.  That is the probability of the edges being matched if they all had the same $x_e$ and $x_e\rightarrow 0$, and it satisfies property \ref{item:twotwo}.  

Now, let us give some intuition about  the second difference between algorithms \PAM and \SPM. We need this to  satisfy property~\ref{item:threethree}. Recall that due to this property we need any edge $e$ to remain unmated with probability at least a constant fraction of $x_e$.  For large values of $x_e$, the transformation function takes care of this property since $x_e-\tilde{x}_e$ will be large as well. However, consider an edge $e=(v,u)$ with a very small $x_e$ which results in  $\tilde{x}_e \approx x_e$. In this case, $v$ sends a proposal to $u$ with probability almost equal to $x_e$. If $e$ is the only edge connected to $u$, this proposal will always be accepted and $e$ joins the matching with probability almost equal to $x_e$ which contradicts property~\ref{item:threethree}. On the other hand, if $u$ has many other edges, the probability of rejecting $v$'s proposal increases. Thus, our idea for satisfying this property is adding dummy edges to the graph to ensure that all the proposals have a constant probability of being rejected. Below we explain what these dummy edges are and how we use them.

\begin{definition}[dummy edges]\label{def:dummy}		
Given a graph $G$ and an LP solution $\b{x}$, let $\ubar\in [0,1]$ be a constant number satisfying $\sum_{e\ni u} x_e \leq \ubar$ for any $u\in B$.  
Dummy edges are edges with   $w_e=0$ and $p_e=1$ which we add to our graph in order to ensure $\sum_{e\ni u} x_e = \ubar$ for all $u\in B$. To achieve this for any vertex $u\in B$ we add a dummy edges $e$ with  $x_e=\ubar -\sum_{e\ni u} x_e$  between $u$ and a degree-one dummy vertex which we add to $A$. 
\end{definition}

It is clear that after adding these  dummy edges, $\b{x}$ is still an optimal solution of the LP. Since the dummy edges weigh zero, they do not change the weight of the matching.  Moreover, due to having $p_e=1$ they do not violate constraint~\eqref{cons:krfjerkf} of the LP either.

\begin{observation}
After adding dummy edges to the graph as described in Definition~\ref{def:dummy}, $\b{x}$ is still an optimal solution of the \ref{LP-Match}.   
\end{observation}

We will talk about the third difference between \SPM and \PAM after stating   \PAM formally. 

\begin{tboxalg2e}{Our $(1-1/e)$-approximation algorithm with specific properties.}\label{algone}
\begin{algorithm}[H]
    \DontPrintSemicolon
    \SetKwProg{Fn}{Function}{:}{}
\Fn{\PAM{$A,B, E, \b x, \ubar$}}{
	For any vertex $u\in B$ ensure $\sum_{e\ni u} x_e=\ubar$ via adding dummy edges. (See Definition~\ref{def:dummy}.) \label{line1}\;
	${\tilde{\b x}} \gets g({\b x}, \ubar)$. \label{line2}\;
	Let $\match=\emptyset$ be a matching of $E$.\;
	Let $\pi$ be a permutation over vertices in $A$ chosen uniformly at random. \;
	\ForEach {$v$ in the order of $\pi$}{
	Let $\pi_v\sim \mathcal{D}^{\tilde{\b x}}_v$ as described by \DIST{$v, \b{x}, \tilde{\b{x}}$}.\label{line:perm}\;
	\ForEach{$e=(v,u)$ permuted in the order of  $\pi_v$}{
        Examine edge $e$.\\
        \If{$e$ is realized}{ Vertex $v$ sends a proposal to $u$.\\
        	If $u$ is unmatched, it accepts the proposal and $e$ joins $\match$.\\
        	Terminate the loop.}
        }
}
\Return $M$.}
\end{algorithm}
\end{tboxalg2e}
\vspace{3mm}
As one can see, this algorithm has two additional lines compared to \SPM, Lines~\ref{line1} and Lines~\ref{line2}. As input, this algorithm gets a bipartite graph $G=(A,B,E)$, a solution of \ref{LP-Match}, $\b{x}$, and a parameter $\ubar$ which should satisfy $\sum_{e\ni u} x_e \leq \ubar$. In Lines~\ref{line1} we add dummy edges to vertices in $B$ to ensure that they all satisfy $\sum_{e\ni u} x_e = \ubar$, and in  Lines~\ref{line2} we transform  $\b{x}$ to another solution of the LP, $\tilde{\b x}$. 

The only other difference between \PAM and \SPM is in the way we find distribution $\mathcal{D}^{\tilde{\b x}}_v$ in Line~\ref{line:perm}. Instead of simply using Lemma~\ref{lemma:dist} to find an arbitrary distribution $\mathcal{D}^{\tilde{\b x}}_v$  proportional to $\tilde{\b{x}}$, we need to construct one with particular features. The goal here is to satisfy property~\ref{item:fourfour}. Recall that due to this property, we want to ensure that conditioned on an edge $e=(v,u)$ being unexamined, its end-points are unmatched with a constant probability. Unfortunately, some proportional distributions result in $\Pr[ v \text{ is unmatched } \mid e \text{ is not examined}\ ]\approx 0$. Below we describe our special $\mathcal{D}^{\tilde{\b x}}_v$ by designing a polynomial time algorithm \DIST for drawing a permutation from it. 
%\vspace{3mm}
%\begin{tboxalg2e}{Constructing a distribution proportional to $\tilde{\b{x}}$ over $E_v$ }
%\begin{algorithm}[H]
%    \DontPrintSemicolon
%\SetKwProg{Fn}{Function}{:}{}
%\Fn{\DIST{$\b{x}, \tilde{\b{x}}$}}{
%   $ \mathcal{D}^{\tilde{\b x}}_v  \gets \mathcal{D}^{\b x}_v$\;
%    \ForEach{$e \in E_v$ }{
%       Let $\mathcal{D}_v \gets \mathcal{D}^{\tilde{\b x}}_v$.\label{line:mwfnjr}\; 
%        \ForEach{$\pi$ with $D_v(\pi) \neq 0$}{
%        Let $\pi_1$ be the prefix of $\pi$ which ends at $e$.\;
%        Let $\pi_2$ be the result of deleting $e$ from permutation $\pi$ .\;
%       $\mathcal{D}^{\tilde{\b x}}_v(\pi) \gets \frac{\tilde{x}_e}{x_e}\cdot \mathcal{D}_v(\pi)$\;
%       $ \mathcal{D}^{\tilde{\b x}}_v(\pi_1) \gets \mathcal{D}_v(\pi_1)+ p_e(1-\frac{\tilde{x}_e}{x_e})\cdot \mathcal{D}_v(\pi)$\;
%         $ \mathcal{D}^{\tilde{\b x}}_v(\pi_2) \gets\mathcal{D}_v (\pi_2)+ (1-p_e)(1-\frac{\tilde{x}_e}{x_e})\cdot \mathcal{D}_v(\pi)$\;
%        }
%        }
%        \Return $\mathcal{D}^{\tilde{\b x}}_v$
% }
%\end{algorithm}
%\end{tboxalg2e}
% 
%\vspace{3mm}

\vspace{3mm}
\begin{tboxalg2e}{Drawing a permutation from $\mathcal{D}^{\tilde{\b x}}_v$}
\begin{algorithm}[H]
    \DontPrintSemicolon
\SetKwProg{Fn}{Function}{:}{}
\Fn{\DIST{$v, \b{x}, \tilde{\b{x}}$}}{
Let $\mathcal{D}^{\b x}_v$ be a distribution over permutations of subsets of $E_v$ from Definition~\ref{def:irifrjf}.\;
   Let $\pi_v\sim  \mathcal{D}^{\b x}_v$\;
   Let $\pi'_v$ be an empty permutation.\;
    \ForEach{$e$ in the order of $\pi_v$}{
    Let $c$ be random number drawn uniformly from $[0,1]$.\;
    \If{$c\leq p_e(1-\frac{\tilde{x}_e}{x_e})$}{
    End the loop.\;
    }
      \If{$p_e(1-\frac{\tilde{x}_e}{x_e})< c\leq p_e(1-\frac{\tilde{x}_e}{x_e})+\frac{\tilde{x}_e}{x_e}$}{
Add  $e$ to the end of permutation $\pi'_v$.\;
}
        }
        \Return $\pi'_v$
 }
\end{algorithm}
\end{tboxalg2e}
 
\vspace{3mm}

We also need to prove that drawing a permutation from $\mathcal{D}^{\tilde{\b x}}_v$  using \DIST{$v, \b{x}, \tilde{\b{x}}$} does not contradict with it being proportional to $\tilde{\b{x}}$. Let $\pi_v\gets $ \DIST{$v, \b{x}, \tilde{\b{x}}$}. By Definition~\ref{def:irifrjf}, $\mathcal{D}^{\tilde{\b x}}_v$ is proportional to $\tilde{\b{x}}$, iff for any edge in $e\in E_v$ probability of $e$ being the first realized one in $\pi_v$ is equal to $\tilde{x}_e$.
Since in \PAM, the inner-loop ends when the first edge from $\pi_v$ is realized, the probability of an edge $e$ being proposed is equal to the probability of $e$ being the first realized edge in $\pi_v$. Therefore, the below lemma proves that $\mathcal{D}^{\tilde{\b x}}_v$  is proportional to $\tilde{x}_e$ as well.  We  defer the proof of the lemma to Appendix~\ref{sec:proofs}. 
%The Let $\pi_v'\gets$ \DIST{$\b{x}, \tilde{\b{x}}$}.  For any edge $e\in E_v$, the probability of $e$ being the first realized edge in $\pi$ is equal to $\tilde{x}_{e}$.

\begin{restatable}{claim}{nierfj}
\label{lemma:nierfj} 
For any edge $e\in E_v$, the probability of $v$ sending a proposal to $u$ in \PAM{} is equal to $\tilde{x}_{e}$.
\end{restatable}

\subsection{Improving Over $(1-1/e)$ Approximation Ratio}

Now we are ready to discuss how we use \PAM to design an algorithm with an improved approximation ratio. We call this algorithm \FNM and formally state it in Algorithm~\ref{alg}. Let us say an edge is available after running \PAM if it remains unexamined with its end-points unmatched. Note that these edges if realized can join the matching. Let us define 
$$r=\sum_{e\in E}x_e w_e\Pr[e \text{ remains available}].$$
If we run \PAM for a second time on these remaining available edges, we find a matching of weight at least $(1-1/e)r$ which we can add  to the matching we found in the first run. Thus if $r$ is at least a constant fraction of the optimal solution of the LP, two runs of \PAM returns a matching with approximation ratio higher than $(1-1/e)$.

 A crucial lemma that we will prove about  \PAM is Lemma~\ref{lemma:ijerg}.  Roughly speaking it states that after we run \PAM once, any edge $e\in E$ will remains available with probability at least $c\cdot(1-\tilde{x}_e/p_e)$ where $c$ is a constant.  (The proof of this lemma strongly relies on properties~\ref{item:threethree} and~\ref{item:fourfour} of \PAM.) Now, let $E^{(\tau)}$ denote the subset of edges  for which $\tilde{x}_e/p_e$ is  no larger that a constant $\tau$. If the contribution of these edges to the optimal solution of the LP is larger than a constant fraction, then the algorithm mentioned above that is running \PAM twice returns a matching with approximation ratio higher than $(1-1/e)$. This is exactly what \FNM does for these graphs. 
 
 The inputs to algorithm \FNM are a graph $G$, the optimal solution of the LP, $\b{x}$, and three parameters $\tau,\ubar,\lambda\in (0,1)$.  As discussed above, given the threshold $\tau$, the algorithm constructs $E^{(\tau)}$ which is the subset of edges for which $\tilde{x}_e/p_e\leq\tau$. Then, it denotes their contribution to the optimal solution of the LP by $\omega$. If $\omega$ is larger than $\lambda$ fraction of the optimal solution, then it runs two rounds of \PAM. Once on the whole graph, and once on the remaining available edges. 
 \vspace{2mm}
\begin{tboxalg2e}{Our algorithm with an approximation ratio larger than $(1-1/e)$.}\label{alg}
\begin{algorithm}[H]
    \DontPrintSemicolon
    \SetKwProg{Fn}{Function}{:}{}
    \Fn{\FNM{$G, \b x, \tau, \ubar, \lambda$}}{
    Let $\match=\emptyset$ be a matching of $E$.\;
    ${\tilde{\b x}} \gets g({\b x}, 1)$. \;
$E^{(\tau)}\gets \{e\in E \mid \tilde{x}_e/p_e\leq \tau\}$.\;
$\omega\gets\sum_{e\in E^{(\tau)}} x_ew_e$.\;
\If{$\omega\geq \lambda(\b{x}\cdot\b{w})$}{

$\match_1\gets$  \PAM{$A, B, E, \b{x}, 1$}.\;
Let $E_2$ be the set of edges in $E$ that are not examined and do not share any end-points with edges in $\match_1$.\;
$\match_2 \gets$ \PAM{$A, B, E_2, \b{x}, 1$}.\;
$\match\gets \match_1\cup \match_2$.\;
}
\Else{
$\match\gets$ \PAM{$A, B, E \setminus E^{(\tau)}, \b{x}, \ubar$}
}
  \Return  $M$.}
\end{algorithm}
\end{tboxalg2e}

\vspace{3mm}
 
 Now, if two runs of \PAM do not achieve an approximation ratio larger than $(1-1/e)$, roughly speaking, it implies that the contribution of edges of $E^{(\tau)}$ to the optimal solution of the LP is very small. Thus, if we remove edges of $E^{(\tau)}$ from the graph and accordingly adjust $\b{x}$, it does not make a significant change in $\b{x}\cdot\b{w}$. We will prove in Claim~\ref{claim:urffddsbgr} that if all the edges have a large $\tilde{x}_e/p_e$, then due to constraint~\eqref{cons:krfjerkf} of the \ref{LP-Match}, $\max_{u\in B} \sum_{e\in E_u} x_e$ is upper-bounded by a constant $\ubar<1$. In this case, the idea is to use property~\ref{item:twotwo} of \PAM, and show that running \PAM{$A, B, E \setminus E^{(\tau)}, \b{x}, \ubar$} achieves an approximation ratio larger than $(1-1/e)$. This is exactly what Algorithm \FNM does in this scenario.

\newcommand{\maintheorem}{Given any bipartite graph $G$ and  $\b{x}$ the optimal solution of \ref{LP-Match} for this graph, it possible to set values of $\tau, \ubar$, and $\lambda$	 such that \FNM{$G,\b{x}, \tau, \ubar, \lambda$} returns an approximate matching with approximation ratio at least $1-1/e+\apximprove$. Moreover, the algorithm is polynomial time in the input size.} 

\begin{restatable}[Main Theorem]{theorem}{maint}\label{thm:main1}
\maintheorem
\end{restatable}

%\begin{tboxalg2e}{}\label{alg}
%\begin{algorithm}[H]
%    \DontPrintSemicolon
%    \SetKwFunction{PAM}{ParitialMatching}
%    \SetKwFunction{FNM}{APXMatching}
%    \SetKwProg{Fn}{Function}{:}{}
%    \Fn{\FNM{$G, \b x, \tau$}}{
%$E^{(\tau)}\gets \{e\in E \mid \tilde{x}_e/p_e\leq \tau\}$.\;
% $M_1\gets$  \PAM{$A, B, E^{(\tau)}, \b{x}, \rho(\tau)$}.\;
%Let $E_2$ be the set of edges in $E$ that are not examined and do not share any end-points with edges in $M_1$.\;
%$M_2 \gets$ \PAM{$A, B, E_2, \b{x}, 1$}.\;
%Let $M\gets M_1\cup M_2$.\;
% \Return  $M$.}
%% 
%% 
%\Fn{\PAM{$A,B, E, \b x, \ubar$}}{
%	For any vertex $u\in B$ ensure $\sum_{e\ni u} x_e=\ubar$ via adding dummy edges. (See Definition~\ref{def:dummy}.)\;
%	${\tilde{\b x}} \gets g({\b x}, \ubar)$ \;
%	Let $M=\emptyset$ be a matching of $E$.\;
%	Let $\pi$ be a permutation over vertices in $A$ chosen uniformly at random. \;
%	\ForEach {$v$ in the order of $\pi$}{
%	Let $\pi_v\sim \mathcal{D}^{\tilde{\b x}}_v$.\;
%	\ForEach{$e=(v,u)$ permuted in the order of  $\pi_v$}{
%        Examine edge $e$.\\
%        \If{$e$ is realized}{ Vertex $v$ sends a proposal to $u$.\\
%        	If $u$ is unmatched, it accepts the proposal and $e$ joins $M$.\\
%        	Terminate the loop.}
%        }
%}
%\Return $M$.}
%\end{algorithm}
%\end{tboxalg2e}

\section{Properties of BaseMatching}\label{section:properties}
In this section, we focus on proving some desirable properties  of \PAM (not only the ones discussed before), which we will later need in our analysis.  

In the lemma stated below, we  provide upper and lower bounds for the probability of any edge joining the matching in  \PAM{$G, \b{x}, \sigma$}. This lemma directly implies the first two properties, \ref{item:twotwo}, and \ref{item:threethree}. Later, we also prove property~\ref{item:fourfour} in Lemma~\ref{lemma:ijerg}.
\begin{lemma}
\label{lem:prm1}

For any edge $e$, we have the following.
\begin{align*}
\frac{(1-e^{-\ubar})x_e}{\ubar} \le \Pr[e \in \match_1] \le \frac{x_e(1+e^{-\ubar})}{2}, \end{align*}
where $\match_1=$\PAM{$G, \b{x}, \sigma$}.

\end{lemma}

\begin{proof}
Before starting the proof, let us mention that for ease of notation, throughout the proof, we will refer to $x_e$ of an edge by its contribution to the solution of the LP (instead of $x_ew_e$).
Note that in algorithm $\match_1=$\PAM{$G, \b{x}, \ubar$}, we use dummy edges to ensure that for any vertex $u\in B$, we have $x_u \gets \sum_{e \ni u} x_e =\ubar.$  Suppose that $e=(v,u)$.   We claim that $\Pr[e \in \match_1]$ is minimized when  except $e$, $u$ has a large number of edges (approaching $\infty$) with the total contribution of $\ubar-x_e$. On the other hand, it is maximized  when $u$ is only incident to one other edge with the contribution of $\ubar-x_e$. To prove our claim, we take an edge $e'=(v', u)\neq e$ and show that replacing it with two edges $e_1'$ and $e_2'$ with $x_{e_1'}=x_{e_2'}= x_e/2$ increases $\Pr[e \in \match_1]$.

Let $\pi$ be the permutation over vertices in $A$ chosen uniformly at random in the algorithm constructing $\match_1$. For any edge $(v'', u)\in E'$, let $s_{v''}$ be a Bernoulli random variable which is equal to one if $v''$ is the first vertex from set $B/\{v'\}$ to send a proposal to $u$. This implies that edge $e$ has a chance of joining the matching only if $s_v=1$. In addition to that, $v'$ should not send a proposal before $v$.  By Claim~\ref{lemma:nierfj}, we know that the probability of any edge $e'$ being proposed by the algorithm is equal to $\tilde{x}_{e'}$.  Thus, we have 

\begin{align*}
 \Pr[e \in \match_1] &=  \Pr\left[v' \text{ does not send a proposal to } u \text{ before } v \mid s_v=1\right]\Pr[s_v=1] \\
 & = \left(1- \Pr\left[\pi(v') < \pi(v) \mid s_v=1\right]\Pr\left[v' \text{ sends a proposal to } u \mid s_v=1\right]\right)\Pr[s_v=1] \\
 & \hspace{-6.3mm}\stackrel{(Claim~\ref{lemma:nierfj})}{=} (1- 0.5\cdot \tilde{x}_{e'})\Pr[s_v=1].
\end{align*}
Now if we replace $e'$ with $e_1'=(v'_1, u)$ and $e_2' = (v'_2, u)$ we get

\begin{align*}
 \Pr[e \in \match_1] &=  \Pr\left[\text{neither } v_1'  \text{ nor }  v_2' \text{ send a proposal to } u \text{ before } v \mid s_v=1\right]\Pr[s_v=1] \\
 & = (1- 0.5\cdot \tilde{x}_{e_1'} - 0.5\cdot \tilde{x}_{e_2'} + 0.25\cdot \tilde{x}_{e_1'}\tilde{x}_{e_2'})\Pr[s_v=1].
\end{align*}
To prove our claim, we need to show 
\begin{align*}
(1- 0.5\cdot \tilde{x}_{e_1'} - 0.5\cdot \tilde{x}_{e_2'} + 0.25\cdot \tilde{x}_{e_1'}\tilde{x}_{e_2'})\Pr[s_v=1]- (1- 0.5\cdot \tilde{x}_{e'})\Pr[s_v=1] \leq 0.	
\end{align*}
Since $\Pr[s_v=1]$ is a non-negative number, we can factor  $0.5\cdot\Pr[s_v=1]$ out and get

\begin{align*}
- \tilde{x}_{e_1'} -  \tilde{x}_{e_2'} + 0.5\cdot \tilde{x}_{e_1'}\tilde{x}_{e_2'} +  \tilde{x}_{e'} \leq 0.	
\end{align*}
Now, plugging in the values of $\tilde{x}_{e'}$, $\tilde{x}_{e_1'}$, and $\tilde{x}_{e_2'}$, we need to prove that the following term is not positive for any $\ubar\in [0,1]$ and $x_e\in [0,\ubar]$.

\begin{align*}
	-2\frac{(e^\ubar-1)(\ubar  -x_e/2) x_e/2}{\ubar (e^\ubar- e^{x_e/2}) }+ 0.5\left(\frac{(e^\ubar-1)(\ubar  -x_e/2) x_e/2}{\ubar (e^\ubar- e^{x_e/2}) }\right)^2 +\frac{(e^\ubar-1)(\ubar  -x_e) x_e}{\ubar (e^\ubar- e^{x_e})}.
\end{align*}
After factoring out $\frac{(e^\ubar-1) x_e}{8\ubar}$ (which is always positive) from all terms, it suffices to prove: 

\begin{align*}
	-\frac{8(\ubar  -x_e/2)}{ (e^\ubar- e^{x_e/2}) }+ \frac{(e^\ubar-1)(\ubar  -x_e/2)^2 x_e}{\ubar (e^\ubar- e^{x_e/2})^2 } +\frac{8(\ubar  -x_e)}{ (e^\ubar- e^{x_e})}\leq 0.
\end{align*}
We do so by plotting this function and observing that it is not positive for any $\ubar\in [0,1]$ and $x_e\in [0,\ubar]$. 
This proves our claim that $\Pr[e \in \match_1]$ is minimized when except $e$, $u$ has a large number of edges (approaching $\infty$) with the total contribution of $\ubar-x_e$. On the other hand, it is maximized  when $u$ is only incident to one other edge with the contribution of $\ubar-x_e$. Now, we are ready to prove the desired lower and upper bounds for $\Pr[e \in \match_1]$. We start by proving the lower-bound.\\
\\
\textbf{Lower-bound:}
\\
Let $k$ denote the number of vertices adjacent to $u$ other than $v$ and define $\delta=\ubar-x_e$. We already proved that $\Pr[e \in \match_1]$ is minimized when $k\rightarrow \infty$ and $x_{e'} = \delta/k$  for any $e'\neq e$. Thus,
\begin{align*}\Pr[e \in \match_1] \\&= \Pr[e \text{ is proposed}]\sum_{i=0}^k\Pr[\text{There are } i \text{ edges before } e \text{ in  } \pi \wedge \text{none of them are proposed}]
\\&\geq g(x_e)\times \lim_{k\rightarrow \infty}\left(\sum_{i=0}^{k}\frac{1}{k+1} \left(1-g\left(\delta/k\right)\right)^i\right) 
\\& = \frac{g(x_e)}{k}\times \lim_{k \rightarrow \infty} \left( \sum_{i=0}^{k} \left(1-\left( \delta/k \right)\right)^i \right)
\\&= g(x_e)\times  \lim_{ k \rightarrow \infty}\left(\sum_{i=0}^{k} \frac{e^{-i\delta/k}}{k}\right) , \end{align*}
where the third equality is due to $\lim_{x\rightarrow 0}g(x)= x$.
Since $k\rightarrow \infty$, we can replace the above summation with an integration as follows:
\begin{align}
\label{eq:integrate}
\sum_{i=0}^{k} \frac{e^{-i\delta/k}}{k}=\int_{0}^{1} e^{-y\delta} dy= \frac{1-e^{-\delta}}{\delta} = \frac{1-e^{-\ubar + x_e}}{\ubar - x_e}.
\end{align}
This gives us:
\begin{align*} \Pr[e \in \match_1]& \geq \frac{(e^\ubar-1)(\ubar  -x_e) x_e}{\ubar (e^\ubar- e^{x_e})}\times \frac{1-e^{-\ubar + x_e}}{\ubar - x_e}
\\ &= \frac{x_e}{\ubar}\times \frac{(e^\ubar-1)(1-e^{-\ubar + x_e})}{e^\ubar-e^{x_e}}
\\ &=\frac{x_e}{\ubar}\times \frac{e^\ubar  - e^{x_e}-1 +  e^{-\ubar + x_e}}{e^\ubar-e^{x_e}}	
\\ &=\frac{x_e}{\ubar}\times \frac{(e^\ubar - e^{x_e}) (1-e^{-\ubar})}{e^\ubar-e^{x_e}}	
\\ &=\frac{x_e(1-e^{-\ubar})}{\ubar}
\end{align*}
 This completes the proof of our desired lower-bound. \\
 \\
 \textbf{Upper-bound:}\\
 Now, we proceed to prove the following upper-bound:
$$\Pr[e \in \match_1] \le \frac{1}{2}(1 + e^\ubar) x_e.$$ 
Recall that we have already proved that this probability is maximized  when $u$ is only incident to one other edge $e'$ with  $x_{e'}=\ubar-x_e$. Hence, we have 
\begin{align*}\Pr\left[e \in \match_1\right] &= \Pr[e \text{ is proposed}]\left(\Pr[\pi(e) < \pi(e') ] + \Pr[\pi(e) > \pi(e')  \wedge e' \text{ is not proposed}]\right)
\\ & = g(x_e) \left( 0.5 + 0.5(1-g(x_{e'})\right)
\\& = g(x_e) \left( 1-0.5g(x_{e'})\right)
\end{align*}
Note that $g(x)/x=\frac{(e^\ubar-1)(\ubar  -x)}{\ubar (e^\ubar- e^x)}$ is a decreasing function of $x$. As such, 
\begin{align*} g(x_e)/x_e\leq  \lim_{x\rightarrow 0}(g(x)/x) =  \lim_{x\rightarrow 0} \frac{(e^\ubar-1)(\ubar  -x)}{\ubar (e^\ubar- e^x)} = \frac{(e^\ubar-1)\ubar  }{\ubar (e^\ubar- 1)}= 1.\end{align*}
Moreover, combining this with the fact that $\lim_{x\rightarrow \ubar} (\ubar-x)/(e^\ubar-e^x) = e^{-\ubar},$ we get
$$g(x_{e'})/x_{e'}\geq \lim_{x\rightarrow \ubar}(g(x)/x) = \lim_{x\rightarrow \ubar} \frac{(e^\ubar-1)(\ubar  -x)}{\ubar (e^\ubar- e^x)} = \frac{(e^\ubar-1)e^{-\ubar}}{\ubar} = \frac{1 - e^{-\ubar}}{\ubar}.$$
As a result, we have 
\begin{align*}\Pr\left[e \in \match_1\right] &= g(x_e) \left( 1-0.5g(x_{e'})\right)
\\&=x_e \frac{g(x_e)}{x_e}\left( 1-0.5x_{e'}\frac{g(\ubar-x_{e'})}{\ubar-x_{e'}}\right)
\\& \leq x_e (1-0.5x_{e'}\frac{1 - e^{-\ubar}}{\ubar})
\\& \leq x_e (1-0.5(\ubar - x_e)\frac{1 - e^{-\ubar}}{\ubar})
\\& \leq x_e (1-0.5(1 - e^{-\ubar}))
\\& \leq x_e (0.5+0.5e^{-\ubar})
\\& \leq \frac{x_e}{2}(1+e^{-\ubar}).
\end{align*}
This proves our desired upper-bound and completes the proof of this lemma.
\end{proof}

We will now prove the most important property of this algorithm. We show that if we run \PAM once, then every edge $e \in E$ remains available with  probability $c \cdot (1- \tilde{x}_e/p_e)$ where $c$ is a constant. After an execution of \PAM, an edge remains available if it is not examined and both of its endpoints are still unmatched. Recall that \PAM guarantees that every edge gets proposed exactly with probability $\tilde{x}_e$. Additionally, the proposal of an edge happens anytime it gets examined and it gets realized. Since the realization of an edge happens independently with probability $p_e$, we can then say that \PAM examines every edge exactly with probability $\tilde{x}_e/p_e$. Thus, an edge $e$ will not get examined in one run of \PAM with probability $(1- \tilde{x}_e/p_e)$.

\begin{fact}
\label{fact:examine}
\PAM{$G, \b{x}, 1$} examines every edge $e \in E$ with probability $\tilde{x}_e/p_e$.
\end{fact}

To show that an edge remains available with probability $c \cdot (1- \tilde{x}_e/p_e)$, we have to show that in case this edge is not examined by \PAM, then with a large constant probability $c$ both of its endpoints remains unmatched. %For the simplicity of the argument, let assume that 
An important property of \PAM is that it ensures every vertex remains unmatched with a large enough probability. This is easy to verify as for a vertex $u \in A \cup B$, Lemma \ref{lem:prm1} implies that every $e \in E_u$ joins the matching with  probability at most $\frac{x_e (1+ e^{-\ubar})}{2}$. Therefore, vertex $u$ gets matched with probability at most $\frac{1+ e^{-\ubar}}{2}$, and it remains unmatched with probability at least $1-\frac{1+ e^{-\ubar}}{2}$. 

Although for every edge $e$ each of its endpoints remains unmatched with a constant probability, we cannot simply argue that $e$ remains available with a large enough probability. The reason is that for an edge $e=(v,u)$, the three events that $e$ is not examined and $v$ and $u$ are unmatched can be potentially correlated. Even though each of these three events happens with a large enough probability, in order to find their joint probability we need a more careful analysis.
\begin{lemma}
\label{lem:r2}
Let $\match_1= $ \PAM{$G, \b{x}, 1$} and $E_2$ be the set of edges in $E$ that are not examined and do not share any endpoints with edges in $\match_1$. Then for any edge $e$ the following holds.
\begin{align*}
\Pr [ e \in E_2 ] \ge  (1- \frac{\tilde{x}_e}{p_e}) \cdot \big(\frac{1-1/e}{2}\big)^2  \,. 
\end{align*}
%For any edge $e$ the probability that $e$ proceeds to the round 2 is at least $(1- \frac{\tilde{x}_e}{p_e}) \cdot (\frac{1- e^{-\sigma}}{2})^2 \cdot e^{-\sigma} \cdot (2 - e^{-\sigma} - \sigma)$.
\end{lemma}
\begin{proof}
Consider an edge $e = (v,u)$, as we discussed in the course of Fact \ref{fact:examine}, \PAM{$G, \b{x}, 1$} examines this edge with  probability $\tilde{x}_e/p_e$. Thus, $e$ does not get examined with probability $\tilde{x}_e/p_e$. In order to prove the lemma we have to show that conditioning on the event that $e$ is not examined, \PAM leaves both $v$ and $u$ unmatched with a large constant probability.

Throughout this proof we use $v \curly u$ to denote the event that $v$ proposes to the $u$ during the execution of \PAM{$G, \b{x}, 1$}, i.e., $v$ examines the edges $(v,u)$ and this edge gets realized. Similarly, we use $v \not \curly u$ to denote the event that $v$ is not proposed to $u$ in the first round. Also, with a slight abuse of notation, we use $v \in \match_1$ to denote the event that $v$ is one of the endpoints of the edges in $\match_1$. Similarly, $v \notin \match_1$ is used to denote the event that $v$ is not any of the endpoints of the edges in $\match_1$.

Using the chain rule, we can expand the probability that $e \in E_2$ as below.

\begin{align*}
&\Pr [e \in E_2]
\\&= \Pr [ \text{$e$ is not examined} \wedge v \notin \match_1 \wedge u \notin \match_1 ] \\
&= \Pr [ \text{$e$ is not examined}] \cdot \Pr[v \notin \match_1 \ |\ \text{$e$ is not examined}] \\
&\quad \cdot\Pr[u \notin \match_1 \ |\ v \notin \match_1 \wedge \text{$e$ is not examined}] 
\end{align*} 

In order to calculate the probability above we take three major steps. We first carefully analyze the probability that $v$ remains unmatched conditioning on the event that $e$ is not examined. In the second step, we show that the event that $u$ remains unmatched by \PAM is positively correlated with the event that $e$ is not examined and $v$ is also unmatched. Specifically, we show that $\Pr[ u \notin \match_1 \ | \ v \notin \match_1 \wedge \text{$e$ is not examined} \ ] \ge \Pr [ u \not\in \match_1]$. Finally, in the last step, we derive the probability that $u$ remains unmatched by \PAM. In this step we show that in the execution of \PAM{$G, \b{x}, 1$}, every vertex of the graph remains unmatched with probability at least $\frac{1-1/e}{2}$. We first focus on the first step. The proof of the lemma is quite technical and it is differed to the appendix.

%Consider an edge $e= (v,u)$, this edge proceeds to the next round if \textit{i)} $e$ is not queried during the first round, and \textit{ii)} Both of its endpoints are unmatched after the first round.  Let $\match_1$ be the matching returned by the algorithm after the first round. Then, the probability that $e$ proceeds to the next round is $\Pr[ u,v \notin \match_1 \text{ and } e \text{ is not queried}]$. We first claim that $\Pr[v \notin \match_1] \ge (1- e^{-\sigma})/2$. Using the Lemma \ref{lem:prm1} we can say that every proposal of $v$ gets rejected with the probability of at least $1-1/2 (1+e^{-\sigma})$. Thus, $v$ remains unmatched during the first round with the probability of at least $(1- e^{-\sigma})/2$. In the next lemma we consider the probability that $e$ is not queried conditioning on the event that $v$ is not matched in the first round.%We now show that conditioning on the event that $v$ is unmatched during the first round, the probability that $u$ is also unmatched is at least 
%$\frac{1-e^{-\bar x}}{\bar x}$.

\begin{restatable}{lemma}{ijerg}
\label{lemma:ijerg}
$\Pr[ v \notin \match_1\ |\ \text{$e$ is not examined }] \ge \frac{1-1/e}{2}$
\end{restatable}

We now focus on the second step of the proof, and we show that the event that $u$ remains unmatched by \PAM is positively correlated with the event that $e$ is not examined and $v$ is also unmatched.
\begin{lemma}
\label{lem:ucorrelation}
$\Pr[ u \notin \match_1 \ | \ v \notin \match_1 \ \wedge \ \text{$e$ is not examined} \ ] \ge \Pr [ u \not\in \match_1]$.
\end{lemma}

\begin{proof}
consider the event that $v$ is not matched by \PAM. There are 3 different cases that can result in $v$ being unmatched. The first case is when $v$ does not proposes to any of its neighbors. It is clear that $v$ remains unmatched in this case. The second case is when $v$ proposes to a vertex $u' \neq u$, however $u'$ has received another proposal before $v$, let say from vertex $v'$. In this case $u'$ gets matched to $v'$, and $v$ remains unmatched. The last case is when $v$ proposes to $u$, however $u$ has got another proposal before $v$. It is also clear that $v$ remains unmatched in this case.
\begin{definition}
The event that $v \in A$ is not matched by \PAM can be partitioned into following disjoint cases.
\begin{itemize}
\item \textbf{Case 1}: $v$ does not propose to any of its neighbors.
\item \textbf{Case 2}: $v$ proposes to a vertex $u' \neq u$, however $u'$ is matched before this proposal.
\item \textbf{Case 3}: $v$ proposes to $u$, and $u$ is matched before this proposal.
\end{itemize}
\end{definition}

Note that Case 3 never happens if we condition on the event that $e$ is not examined. Thus,
\begin{align}
\label{eq:case3}
\Pr [ \text{Case 3} \  | \ v \notin \match_1 \ \wedge \ \text{$e$ is not examined}] =0 \,.
\end{align}
 
 We show that for the first two cases, the two events that $u$ remains unmatched by \PAM is positively correlated with the event that $e$ is not examined and $v$ is also unmatched. This immediately implies the following and proves the lemma.
 
 \begin{align*}
 & \Pr[ u \notin \match_1 \ | \ v \notin \match_1 \ \wedge \ \text{$e$ is not examined}] \\
 &= \sum_{i=1}^3 \Pr[ u \notin \match_1 \ | \ \text{Case } i \ \wedge \ \text{$e$ is not examined}] \cdot \Pr [ \text{Case } i \  | \ v \notin \match_1 \ \wedge \ \text{$e$ is not examined}] \\
  &= \sum_{i=1}^2 \Pr[ u \notin \match_1 \ | \ \text{Case } i \ \wedge \ \text{$e$ is not examined}] \cdot \Pr [ \text{Case } i \  | \ v \notin \match_1 \ \wedge \ \text{$e$ is not examined}]
& \text{By (\ref{eq:case3})} \\
& \ge \sum_{i=1}^2 \Pr[ u \notin \match_1 ] \cdot \Pr [ \text{Case } i \  | \ v \notin \match_1 \ \wedge \ \text{$e$ is not examined}] \\
& \qquad \text{(Assuming the positive correlation.)} \\ 
& = \Pr[ u \notin \match_1 ] \cdot  \sum_{i=1}^2 \Pr [ \text{Case } i \  | \ v \notin \match_1 \ \wedge \ \text{$e$ is not examined}] \\
& = \Pr[ u \notin \match_1 ]  \,.\\
& \qquad \text{(Since Case 3 never happens)} \\ 
 \end{align*}

Consider the Case 1 when $v$ proposes to none of its neighbors. In this case for every incident edge $e'$ of $v$, either $e'$ is not examined, or it is not realized. In other words, we have $v \not \curly u'$ for every $u' \in B$. In this case we have
\begin{align*}
&\Pr[ u \notin \match_1 \ | \ \text{Case } 1 \ \wedge \ \text{$e$ is not examined}] \\
&=\Pr[ u \notin \match_1 \ |\ \text{$v$ does not proposes to neighbors} \ \wedge \ \text{$e$ is not examined} ] \\
&=\Pr[ u \notin \match_1 \ | \bigwedge_{ u' \in B} v \not \curly u'  \ \wedge \ \text{$e$ is not examined} ] \\
&= \Pr[ \bigwedge_{ v' \in A} v' \not \curly u \ |  \bigwedge_{ u' \in B} v \not \curly u'  \ \wedge \ \text{$e$ is not queried}] \,,
\end{align*}
where the equality above is derived from the fact that $u$ remains unmatched if and only if it gets no proposal from its neighbors. Considering the probability above, the two events that $v' \not \curly u$ and $v \not \curly u'$ are independent if $v' \neq v$. The is because every vertex in $A$ proposes to its neighbors independent of other vertices in $A$. Due to the same reason, an event $v' \not \curly u$ is independent from the event  that $e$ is not examined if $v' \neq v$. Thus,
\begin{align*}
&\Pr[ u \notin \match_1 \ | \bigwedge_{ u' \in B} v \not \curly u'  \ \wedge \ \text{$e$ is not examined}] \\
&= \Pr[ \bigwedge_{ v' \in A\setminus \{v\}} v' \not \curly u] \cdot \Pr[ v \not \curly u \ |  \bigwedge_{ u' \in B} v \not \curly u'  \ \wedge \ \text{$e$ is not examined}] \\
&=\Pr[ \bigwedge_{ v' \in A\setminus \{v\}} v' \not \curly u] \\ & (\text{Since $\Pr[ v \not \curly u \ |  \bigwedge_{ u' \in B} v \not \curly u'  \ \wedge \ \text{$e$ is not examined}]=1$}) \\
&\ge \Pr[ \bigwedge_{ v' \in B} v' \not \curly u] = \Pr [ u \notin \match_1] \,. \\
& (\text{$u$ remains unmatched if it gets no proposal})
\end{align*}
This shows the positive correlation for Case 1.

We now show the positive correlation for Case 2. In this case $v$ remains unmatched in \PAM because $v$ proposes to a vertex $u' \neq u$, however $u'$ is already matched before $v$. For a random permutation $\pi$ over the vertices in $A$, we use $\pi_{<v}$ to denote the set of vertices that are earlier than $v$ in the permutation, i.e., $\pi_{<v}= \{v' \in A : \pi(v') < \pi(v)\}$. Then, we can say that Case 2 happens if $v$ proposes to a vertex $u'$, and at least one other vertex $v' \in \pi_{<v}$ also proposes to $u'$. We use $X_{v,u'}$ to denote the event that $v$ proposes to $u'$, however at least one other vertex $v' \in \pi_{<v}$ has already proposed to $u'$. Thus, the probability that $u$ remains unmatched conditioned that Case 2 happens is as follows.
\begin{align}
&\Pr[ u \notin \match_1 \ | \ \text{Case } 2 \ \wedge \ \text{$e$ is not examined}] \nonumber \\
& = \sum_{u' \neq  u} \Pr[ u \notin \match_1 \ | \ X_{v,u'} \ \wedge \ \text{$e$ is not examined}] \cdot \Pr[ X_{v,u'} |  \ \text{Case } 2 \ \wedge \ \text{$e$ is not examined}] \nonumber \\
& = \sum_{u' \neq  u} \Pr[ \bigwedge_{v' \in A} v' \not \curly u \ | \ X_{v,u'} \ \wedge \ \text{$e$ is not examined}] \cdot \Pr[ X_{v,u'} |  \ \text{Case } 2 \ \wedge \ \text{$e$ is not examined}] \nonumber \\
\label{eq:case2}
& = \sum_{u' \neq  u} \bigg(\Pr[ X_{v,u'} |  \ \text{Case } 2 \ \wedge \ \text{$e$ is not examined}] \cdot \prod_{v' \in A} \Pr[ v' \not \curly u \ | \ X_{v,u'} \ \wedge \ \text{$e$ is not examined}] \bigg) \,.\\
&\text{(Since proposals of vertices in $A$ are independent)}  \nonumber
\end{align}
% In other words, we have $v \curly u'$, $v' \curly u$, and $v' \preceq v$. In this case we have

In the claim below we show that for each $u' \neq u$ and $v' \in A$, we have $\Pr[v' \not \curly u] \le \Pr[ v' \not \curly u \ | \ X_{v,u'} \ \wedge \ \text{$e$ is not examined} \ ]$. The proof of the claim is available in appendix.

\begin{restatable} {claim} {clmvpnotu}
%\begin{claim}
For an edge $e=(v,u)$,  and every $u' \in B$ and $v' \in A$ such that $u \neq u'$, we have $$\Pr[v' \not \curly u] \le \Pr[ v' \not \curly u \ | \ X_{v,u'} \ \wedge \ \text{$e$ is not examined} \ ] \,.$$
%\end{claim}
\end{restatable}

By the claim above and (\ref{eq:case2}) we then have
\begin{align*}
&\Pr[ u \notin \match_1 \ | \ \text{Case } 2 \ \wedge \ \text{$e$ is not examined}] \\
& = \sum_{u' \neq  u} \bigg(\Pr[ X_{v,u'} |  \ \text{Case } 2 \ \wedge \ \text{$e$ is not examined}] \cdot \prod_{v' \in A} \Pr[ v' \not \curly u \ | \ X_{v,u'} \ \wedge \ \text{$e$ is not examined}] \bigg) \\
&\ge \sum_{u' \neq  u} \bigg(\Pr[ X_{v,u'} |  \ \text{Case } 2 \ \wedge \ \text{$e$ is not examined}] \cdot \prod_{v' \in A} \Pr[ v' \not \curly u ] \bigg) \\
&=  \sum_{u' \neq  u} \bigg(\Pr[ X_{v,u'} |  \ \text{Case } 2 \ \wedge \ \text{$e$ is not examined}] \cdot \Pr[ \bigwedge_{v' \in A} v' \not \curly u ] \bigg) \\
&\text{(Since proposals of vertices in $A$ are independent)}  \\
&= \Pr[ \bigwedge_{v' \in A} v' \not \curly u ] = \Pr[ u \notin \match_1] \,.
\end{align*}
This shows the positive correlation for the second case, and proves the lemma.
% We also use $v \not \curly U$ to denote the event that $v$ does not propose to any of the vertices $U$ during the first round. Note that in this case some of the edges incident to $v$ might be queried during the first round, however it is guaranteed that they are not realized.

%Now consider 

\end{proof}

We now focus on the last step of the proof and we show that if we run \PAM{$G, \b{x}, 1$}, every vertex $u \in B$ remains unmatched with probability at least $\frac{1-1/e}{2}$. 

\begin{lemma}
\label{lem:uunmatched}
For every vertex $u \in B$, the following holds.
\begin{align*}
\Pr[u \notin \match_1] \ge \dfrac{1-1/e}{2}
\end{align*}
\end{lemma}
\begin{proof}
Since at most one of the incident edges of $u$ are in $\match_1$, the following holds.
\begin{align*}
\Pr [ u \in \match_1] &= \sum_{e \ni u} \Pr [ e \in \match_1] \\
&\le  \sum_{e \ni u} \frac{x_e (1+1/e)}{2} & \text{By Lemma \ref{lem:prm1} with \ubar=1} \\
&= \frac{1+1/e}{2} \cdot \sum_{e \ni u} x_e \\
& \le \frac{1+1/e}{2} \,.  
\end{align*}
Thus, the probability that $u$ remains unmatched is at least $1-\frac{1+1/e}{2}= \frac{1-1/e}{2}$.
\end{proof}
Now by Lemma \ref{lemma:ijerg}, \ref{lem:ucorrelation} and \ref{lem:uunmatched} we have
\begin{align*}
&\Pr [e \in E_2]
\\&= \Pr [ \text{$e$ is not examined}] \cdot \Pr[v \notin \match_1 \ |\ \text{$e$ is not examined}] \\
&\quad \cdot\Pr[u \notin \match_1 \ |\ v \notin \match_1 \wedge \text{$e$ is not examined}]  \\
&\ge (1- \frac{\tilde{x}_e}{p_e}) \cdot \big(\frac{1-1/e}{2}\big)^2 \,.
\end{align*}

\end{proof}

\section{Analysis of APXMatching}\label{section:approx-ratio}
Consider the LP of the problem, we show that there is a adaptive algorithm that given any feasible solution of the problem achieves a matching with the weight of at least $(1-1/e + \apximprove) \cdot \b{x} \cdot \b{w}$ in expectation. 

According to the lemma \ref{lem:prm1}, during the single run of our algorithm according to the fractional solution $\tilde{\b{x}}$ it is guaranteed that the algorithm matches every edge in the graph with  probability  at least $(1-1/e)$, thus it provides us $(1-1/e)$ approximation. Our idea to beat this approximation factor is to  leverage Lemma \ref{lem:r2} and show that every edge proceeds to the second round with a large probability. Recall that Lemma \ref{lem:r2} shows that if we run  \PAM{$G, \b{x}, 1$} every edge proceeds to the second round with probability at least $(1- \frac{\tilde{x}_e}{p_e}) \cdot \big(\frac{1-1/e}{2}\big)^2$. This however, cannot guarantee that every edge proceeds to the second round with a non-zero constant probability. For example, consider an edge with the $\tilde{x}_e=p_e$. Then our algorithm queries this edge with probability $\frac{\tilde{x_e}}{p_e}=1$. Therefore, this edge never proceeds to the second round. This is basically the same reason that the factor of $(1- \frac{\tilde{x}_e}{p_e})$ appears in the probability bound of Lemma \ref{lem:r2}.

Let $\tau$ be a constant in which we determine later. We call edges with $\frac{\tilde{x}_e}{p_e}\ge\tau$, the  \textit{heavy edges}. During the first round of \PAM, the algorithm examines every edge of the graph with probability $\frac{\tilde{x_e}}{p_e}$. Considering the heavy edges, the value of $\frac{\tilde{x}_e}{p_e}$ is large for them, and these edges have a higher chance to get examined during the first round of algorithm. On the other hand, these edges are less likely to be available for the second round. Due to this asymmetry between the edges of the graph, we divide them in to two groups. For each vertex $u$, we use $H_u$ to denote the set of heavy edges incident to $u$. Specifically, $H_u= \{e = (v,u)\ | \ \frac{\tilde{x}_e}{p_e}\ge\tau\}$. We also use $H$ to denote all heavy edges, i.e., $H= \bigcup_{u \in U} H_u$(Note that $H= E \setminus E^{(\tau)}$). We also use $\wdotx_H= \sum_{e \in H} w_e \cdot x_e$ to denote the weighted contribution of the edges in $H$. Suppose that we run the algorithm for two rounds with the parameter $\sigma=1$, in the following lemma we study the approximation guarantee of our algorithm.

\begin{lemma}
\label{lem:tworounds}
Let $\match_1 = $ \PAM{$G, \b{x}, 1$}, $E_2$ be the set of edges in $E$ that are not examined and do not share any endpoints with edges in $\match_1$. Also, let $\match_2=  $ \PAM{$A, B, E_2, \b{x}, 1$}, and $\match = \match_1 \cup \match_2$.  Then,
\begin{align*}
\frac{\E[\W{\match}]}{\opt} \ge  (1-1/e) + \frac{(1-1/e)^3}{4} \cdot (1- \frac{\wdotx_H}{\optlp}) \cdot  (1-\tau) \,,
\end{align*}
where $\optlp = \b{x} \cdot \b{w}$ is the solution of the LP.
\end{lemma}
\begin{proof}
Consider the first round \PAM, due to Lemma \ref{lem:prm1}, when $\sigma=1$, the algorithm  matches every edge $e$ with probability at least $ (1-1/e) \cdot x_e$. Thus,
\begin{align}
\E[\W{\match_1}] &= \sum_{e \in E} w_e \cdot \Pr[ e \in \match_1] \nonumber \\
& \ge \sum_{e \in E} w_e \cdot (1-1/e) \cdot x_e & \text{By Lemma \ref{lem:prm1}} \nonumber\\
\label{eq:em1}
& = (1-1/e) \cdot \sum_{e \in E} w_e \cdot x_e = (1-1/e) \cdot \optlp \,.
\end{align}
Therefore, the expected size of the matching in the first round of the algorithm is at least $(1-1/e) \cdot \opt$.
% Therefore, we can say that the algorithm matches every edge in its first round of execution with the probability of at least $(1-1/e)$, and it finds a $(1-1/e)$ approximation of the problem.

Now consider the second round when \PAM creates the matching $\match_2$. Using the same argument we can say that if an edge $e$ proceeds to the second round, then it gets matched with probability at least $(1-1/e) \cdot x_e$. Thus,
\begin{align}
\label{eq:prm2}
\Pr[ e \in \match_2 ] \ge \Pr [ e \text{ proceeds to the round 2}] \cdot (1-1/e) \cdot x_e \,.
\end{align} 
Therefore,
\begin{align*}
&\E[\W{\match_2}] = \sum_{e \in E} w_e \cdot \Pr[ e \in \match_2]  \\
&\ge  \sum_{e \in E} (1-1/e) \cdot w_e \cdot x_e \cdot  \Pr [ e \text{ proceeds to the round 2}] & \text{By (\ref{eq:prm2})} \\
&\ge  \sum_{e \in E} (1-1/e) \cdot w_e \cdot x_e \cdot (1- \frac{\tilde{x}_e}{p_e}) \cdot \big(\frac{1-1/e}{2}\big)^2  &\text{By Lemma \ref{lem:r2}} \\
&= \sum_{e \in E} \frac{(1-1/e)^3}{4} \cdot w_e \cdot x_e \cdot (1- \frac{\tilde{x}_e}{p_e})\\
&=   \frac{(1-1/e)^3}{4} \cdot \sum_{e \in E} w_e \cdot x_e \cdot (1- \frac{\tilde{x}_e}{p_e}) \,.
\end{align*}

Now consider an edge $e \notin H$. For this edge, we have $\frac{\tilde{x}_e}{p_e} \le \tau$.  Using this and the inequality above we have,
\begin{align}
&\E[\W{\match_2}] \ge \frac{(1-1/e)^3}{4} \cdot \sum_{e \in E \setminus H} w_e \cdot x_e \cdot (1- \frac{\tilde{x}_e}{p_e}) \nonumber \\
& \ge \frac{(1-1/e)^3}{4} \cdot \sum_{e \in E \setminus H} w_e \cdot x_e \cdot (1- \tau) \nonumber \\
& =  (1-\tau) \cdot \frac{(1-1/e)^3}{4} \cdot \sum_{e \in E \setminus H} w_e \cdot x_e \nonumber \\
& =   (1- \tau)  \cdot \frac{(1-1/e)^3}{4} \cdot \big(\sum_{e \in E} w_e \cdot x_e - \sum_{e \in H} w_e \cdot x_e \big) \nonumber \\
\label{eq:em2}
& = (1- \tau)  \cdot \frac{(1-1/e)^3}{4} \cdot  (\optlp - \wdotx_H)
\end{align}
Now consider the matching $\match$ returned by the algorithm after the first two rounds. We have $\match = \match_1 \cup \match_2$. Thus, we have
\begin{align*}
&\E[\W{\match}] = \E[\W{\match_1}]+ \E[\W{\match_2}] & \text{By linearity of expectation} \\
&\ge (1-1/e) \cdot \optlp +  (1- \tau)  \cdot \frac{(1-1/e)^3}{4} \cdot (\optlp - \wdotx_H) & \text{By (\ref{eq:em1}) and (\ref{eq:em2})}
\end{align*}
Using the inequality above and the fact that $\opt \le \optlp$, we have
\begin{align*}
\frac{\E[\W{\match}]}{\opt} \ge \frac{\E[\W{\match}]}{\optlp} \ge  (1-1/e) +  (1- \tau)  \cdot \frac{(1-1/e)^3}{4} \cdot (1 - \frac{\wdotx_H}{\optlp}) \,.
\end{align*}
This proves the lemma.
\end{proof}

Suppose that $\wdotx_H \le (1-\Omega(1)) \optlp$. Then, lemma above implies that the matching returned using two round of \PAM, has an expected weight of at least $\opt \cdot (1-1/e) + \opt  \cdot (1- \tau)  \cdot \frac{(1-1/e)^3}{4} \cdot \Omega(1)$. Thus, the expected weight of the matching is at least $\opt (1-1/e +\Omega(1))$, and the approximation factor of the algorithm is better than $(1-1/e)$ by a constant factor. This implies that the only case that this algorithm cannot beat $(1-1/e)$ approximation ratio is when $\wdotx_H = \optlp (1-o(1))$. We show that in this case we can run Algorithm \PAM for only a single round on edges in $H$, it beats the $(1-1/e)$ approximation ratio.

In this case the main idea is to show that $\sum_{e \in H_u} x_e$ is much smaller than $1$ for all vertices $u \in U$. Therefore, we can run \PAM with the $\sigma <1$ and beat the approximation ratio of $(1-1/e)$ in a single round. 
Let $\tau > 1-1/e$, we first show that for every edge $e \in H$ we have $x_e <1$. Recall that during the first round of the algorithm we set $\tilde{x}_e= g(x_e,1)$. Then, we have
\begin{align*}
\frac{\tilde{x}_e}{p_e} &= \frac{g(x_e,1)}{p_e} \\
& \le \frac{g(x_e,1)}{x_e} \,. & \text{Since $p_e \ge x_e$}
\end{align*}
Furthermore, for edges in $H$ we have $\frac{\tilde{x}_e}{p_e} \ge \tau$. By combining this with the inequality above we can say that for all edges in $H$ we have 
\begin{align}
\label{eq:xeovere}
\frac{g(x_e,1)}{x_e} \ge \tau \,.
\end{align}
We claim that $\frac{g(x_e,1)}{x_e}$ is strictly decreasing with respect to $x_e$. The proof of the claim is differed to the appendix.
\begin{restatable}{claim}{gxderiv}
\label{clm:gdev}
Function $\frac{g(x,\ubar)}{x}$ is strictly decreasing with respect to $x$ for $0 \le x \le \ubar$.
\end{restatable}
Using the claim above we have $\frac{g(x,1)}{x} \le \frac{g(1,1)}{1} = (1-1/e)$. Using this and the fact that $\frac{g(x,1)}{x}$ is decreasing w.r.t. $x$, we can say for any $\tau \ge (1-1/e)$, there exists a threshold denoted by $\gm(\tau)$ such that for any $x_e \ge \gm(\tau)$ we have $\frac{g(x_e,1)}{x_e} \le \tau$. 

It is now easy to show that for all edges in $e \in H$, we have $x_e \le \gm(\tau)$. For the sake of contradiction suppose that $x_e > \gm(\tau)$. We then get the following
\begin{align*}
&\frac{g(x_e,1)}{x_e} < \frac{g(\gm(\tau),1)}{\gm(\tau)} & \text{Since $\frac{g(x_e,1)}{x_e}$ is decreasing and $x_e > \gm(\tau)$} \\
& \le \tau \,. & \text{By definition of $\gm(\tau)$}
\end{align*}
This however contradicts (\ref{eq:xeovere}). Thus, for all edges in $H$ we have $x_e \le \gm(\tau)$.
%By substituting the function $g$, we get the following inequality for all edges in $H$.
%\begin{align*}
%\frac{g(x_e,1)}{x_e} \ge \tau \Longleftrightarrow \frac{(1-1/e)(1-x)}{1-e^{x-1}} \ge \tau
%\end{align*}

We are now ready to show that for a proper choice of $\tau$, the summation $\sum_{e \in H_u} x_e$ is much smaller $1$ for all vertices $u \in U$. Consider a vertex $u$, then due to constraints of the LP we have
\begin{align*}
\sum_{e \in H_u} x_e \le \Pr[ \text{an edge in $H_u$ exists}] \,.
\end{align*}
Since each edge in $H_u$ realizes independently with probability of $p_e$, we have
\begin{align}
\sum_{e \in H_u} x_e &\le \Pr[ \text{an edge in $H_u$ exists}] \nonumber \\
& = 1- \prod_{e \in H_u} (1- p_e) \nonumber \\
\label{eq:sumxe}
& \le 1- \prod_{e \in H_u} (1- \frac{x_e}{\tau}) \,. & \text{Since $\frac{x_e}{p_e} \ge \frac{\tilde{x}_e}{p_e} \ge \tau$ for edges in $H_u$}
\end{align}
For each edge $e \in H_u$, let $a_e = \frac{x_e}{\tau}$. By the inequality above we then have, 
$$\sum_{e \in H_u} x_e  \le 1- \prod_{e \in H_u} (1- a_e) \,.$$
 By the definition of $a_e$, we also have that $\sum_{e \in H_u} a_e = \frac{\sum_{e \in \tau} x_e}{\tau}$. We then use the lemma below to bound $\prod_{e \in H_u} (1- a_e)$. The proof of the lemma is differed to the appendix.

\begin{restatable}{lemma}{minsum}
\label{lem:minsum}
Let $0 \le c \le 1$, and $a_1, a_2, \cdots, a_k \in [0,c]^k $ be $k$ non-negative real numbers bounded by $c$. Further assume that $\sum_{i=1}^k a_i =s$ where $s \ge 0$ is a fixed constant. Then, the following always holds.
\begin{align*}
\prod_{i=1}^k (1- a_i) \ge (1-c)^{\frac{s}{c}} \,.
\end{align*} 
\end{restatable}

Recall that as we discussed we can assume that $x_e$ is bounded by $\gm(\tau)$, therefore each $a_e$ is bounded by $\frac{\gm(\tau)}{\tau}$. It can be verified that for $\tau \ge 3/4$, we always have $\frac{\gm(\tau)}{\tau} <1$, thus each $a_e$ a positive number  and it is less than $1$.

\begin{restatable}{claim}{gmovertau}
\label{clm:gmovertau}
For any $3/4 \le \tau \le 1$, we have $\frac{\gm(\tau)}{\tau} <1$.
\end{restatable}

 We can then set the parameter $c= \frac{\gm(\tau)}{\tau}$ in Lemma \ref{lem:minsum} to get the following inequality.
\begin{align*}
\prod_{e \in H_u} (1- a_e) &\ge \bigg(1-\frac{\gm(\tau)}{\tau}\bigg)^{\bigg(\frac{\sum_{e\in H_u}{a_e}}{\frac{\gm(\tau)}{\tau}}\bigg)} \\
&= \bigg(1-\frac{\gm(\tau)}{\tau}\bigg)^{\big(\frac{\sum_{e\in H_u}{x_e}}{\gm(\tau)}\big)} \,. & \text{Since $\sum_{e \in H_u} a_e = \frac{\sum_{e \in \tau} x_e}{\tau}$}  
\end{align*}
Using the inequality above and (\ref{eq:sumxe}) we get the following constraint for the summation $\sum_{e \in H_u} x_e$.
\begin{align}
\label{const:sumh1}
\sum_{e \in H_u} x_e \le 1- \prod_{e \in H_u} (1- a_e) \le 1- \bigg(1-\frac{\gm(\tau)}{\tau}\bigg)^{\big(\frac{\sum_{e\in H_u}{x_e}}{\gm(\tau)}\big)}  \,.
\end{align}

For a $\tau \ge 3/4$, let $\rho(\tau)$ be the largest positive real value $x \in (0,1]$ such that $x \le 1-(1-\frac{\gm(\tau)}{\tau})^{\big(\frac{x}{\gm(\tau)}\big)}$. It can be shown that $\rho(\tau)$ is strictly less than $1$ for $\tau \ge 3/4$.
\begin{restatable}{claim}{rholessone}\label{claim:urffddsbgr}
For any $3/4 \le \tau < 1$, we have $\rho(\tau) <1$.
\end{restatable}
 We can then re-write the constraint (\ref{const:sumh1}) as below.
\begin{align}
\label{const:sumh2}
\sum_{e \in H_u} x_e \le \rho(\tau) \,.
\end{align}

The constraint above shows that if we only consider the edges in $H$, then the summation of $x_e$ for all edges incident to a vertex $u \in U$ is bounded by $\rho(\tau)$. Therefore, we can run Algorithm \PAM by setting $\sigma=\rho(\tau)$. Let $\match$ be the matching returned by the algorithm, then we have

\begin{align*}
\E[\W{\match}] &= \sum_{e \in H} w_e \cdot \Pr[ e \in \match] \nonumber \\
& \ge \sum_{e \in H} w_e \cdot \frac{1-e^{-\rho(\tau)}}{\rho(\tau)} \cdot x_e & \text{By Lemma \ref{lem:prm1}} \nonumber\\
& = \frac{1-e^{-\rho(\tau)}}{\rho(\tau)} \cdot \sum_{e \in H} w_e \cdot x_e = \frac{1-e^{-\rho(\tau)}}{\rho(\tau)} \cdot  \wdotx_H \,.
\end{align*}

Thus, in this case the approximation factor of the algorithm is at least $\frac{1-e^{-\rho(\tau)}}{\rho(\tau)} \cdot  \frac{\wdotx_H}{\opt} \ge \frac{1-e^{-\rho(\tau)}}{\rho(\tau)} \cdot  \frac{\wdotx_H}{\optlp}$.  This gives us the following lemma.

\begin{lemma}
\label{lem:oneroundheavy}
let $\match=  $ \PAM{$A, B, E \setminus E^{(\tau)}, \b{x}, \rho(\tau)$}, Then,
\begin{align*}
\frac{\E[\W{\match}]}{\opt} \ge  \frac{1-e^{-\rho(\tau)}}{\rho(\tau)} \cdot  \frac{\wdotx_H}{\optlp} \,,
\end{align*}
\end{lemma}

Let $\wdotx= \sum_{e \in E^(\tau)} x_e w_e$, then we have $\wdotx_H= \optlp-\wdotx$. Suppose that we run the algorithm mentioned in Lemma  \ref{lem:tworounds} if $\wdotx \ge \lambda \cdot \optlp$ where $\lambda$ is a constant we optimize later. In this case the approximation ratio of the algorithm is at least
\begin{align}
&(1-1/e) + \frac{(1-1/e)^3}{4} \cdot (1- \frac{\wdotx_H}{\optlp}) \cdot  (1-\tau) \nonumber \\
&=(1-1/e) +  \frac{(1-1/e)^3}{4} \cdot (\frac{\wdotx}{\optlp}) \cdot  (1-\tau) \nonumber \\
\label{eq:algapx1}
&\ge (1-1/e) + \frac{(1-1/e)^3}{4} \cdot \lambda \cdot  (1-\tau) 
\end{align} 
Also, suppose that we run the algorithm used to derive Lemma \ref{lem:oneroundheavy} if $\wdotx < \lambda \cdot \optlp$. In this case the approximation ratio is
\begin{align}
\label{eq:algapx2}
\frac{1-e^{-\rho(\tau)}}{\rho(\tau)} \cdot  \frac{\wdotx_H}{\optlp}  = \frac{1-e^{-\rho(\tau)}}{\rho(\tau)} \cdot  (1-\frac{\wdotx}{\optlp}) \ge  \frac{1-e^{-\rho(\tau)}}{\rho(\tau)} \cdot (1-\lambda)
\end{align}
By (\ref{eq:algapx1}) and (\ref{eq:algapx2}) for a fixed parameter $\lambda$ the approximation ratio of \FNM is
\begin{align}
\label{eq:approxlambda}
\min\{(1-1/e) + \frac{(1-1/e)^3}{4} \cdot \lambda \cdot  (1-\tau), \frac{1-e^{-\rho(\tau)}}{\rho(\tau)} \cdot (1-\lambda)\} \,.
\end{align}
In order to optimize the performance of our algorithm we set $\tau = 0.8723$. We then claim that for this choice of the parameter, $\rho(\tau) \le 0.5303$. The proof of the claim is available in the appendix.

\begin{restatable}{claim}{rhofinal}
For $\tau= 0.8723$, we have $\rho(\tau) \le 0.5303$.
\end{restatable}

Note that the function $\frac{1-e^{-x}}{x}$ is a decreasing function w.r.t. $x$ and we have $\rho(\tau) \le 0.5303$. As a result,
\begin{align}
\label{eq:bestalg2}
\frac{1-e^{-\rho(\tau)}}{\rho(\tau)} \ge \frac{1-e^{-0.5303}}{0.5303} > 0.7761
\end{align}
Thus, the approximation ratio of the algorithm is equal to 
\begin{align*}
&\min\{(1-1/e) + \frac{(1-1/e)^3}{4} \cdot \lambda \cdot  (1-\tau), \frac{1-e^{-\rho(\tau)}}{\rho(\tau)} \cdot (1-\lambda)\}  & \text{By (\ref{eq:approxlambda})}\\
&\ge \min\{(1-1/e) + \frac{(1-1/e)^3}{4} \cdot \lambda \cdot  (1-\tau), 0.7761 \cdot (1-\lambda)\} \,,
\end{align*}
This implies that for $\lambda=0.1837$, the approximation ratio of our algorithm is at least $0.63353> 1-1/e+\apximprove$. This proves our main theorem (Theorem \ref{thm:main1}) by the choice of variables $\tau=0.8732$, $\sigma = 0.5303$ and $\lambda= 0.1873$.  

\maint*
\begin{proof}
As proved in this section, \FNM achieves the approximation ratio of $0.63353> 1-1/e+\apximprove$ when $\tau=0.8732$, $\sigma = 0.5303$ and $\lambda= 0.1873$. Also, all the decisions  of the algorithm are performed in polynomial time. 
\end{proof}

\newpage

\bibliographystyle{plain}
\bibliography{refs}

\newpage

\appendix

\section{Omitted Proofs}\label{sec:proofs}
\subsection{Omitted Proofs of Section 4}
\xless*
\begin{proof}
We show that $\frac{g(x,\ubar)}{x} \le 1$. We have the following.
\begin{align*}
\frac{g(x,\ubar)}{x} = \frac{(e^\ubar-1)(\ubar  -x)}{\ubar (e^\ubar- e^x)} 
\end{align*}
Since $\frac{g(x,\ubar)}{x}$ is decreasing with respect to $x$ (see Claim \ref{clm:gdev}), it gets maximized when $x$ approaches $0$ in which we have
\begin{align*}
\frac{g(x,\ubar)}{x} \le \lim_{x \rightarrow 0} \frac{g(x,\ubar)}{x} =  \frac{(e^\ubar-1)(\ubar )}{\ubar (e^\ubar- 1)} =1 \,. 
\end{align*}
\end{proof}
\nierfj*
\begin{proof}
Vertex $v$ send a proposal by drawing a random permutation from \DIST{$v, \b{x}, \tilde{\b{x}}$}. It then examines the edges in the same order of the permutation and proposes the first realized edge. Therefore, instead of showing that $v$ sends a proposal to $u$ with the probability $\tilde{x}_e$, we can show that if we draw a random permutation $\pi'_v$ from \DIST, and examine the edges in the same order, then the edge $(v,u)$ would be the first realized edge with the probability of $\tilde{x}_e$.

If we draw a random permutation $\pi_v$ from $\mathcal{D}^{\b x}_v$, then for each $e \in E_v$ the probability that $e$ is the first edge that realizes in $\pi_v$ is $x_e$. Consider any fixed permutation $\bar{\pi}_v$, and let $\bar{\pi}'_v$ be a random permutation that \DIST returns when in Line 2 of \DIST $\pi_v=\bar{\pi}_v$. For any edge $e$ let $q_e$ be the probability that $e$ is the first edge that gets realized in $\bar{\pi}_v$, and $q'_e$ be the probability that $e$ is the first edge that gets realized in  $\bar{\pi}'_v$. We then show that $q'_e = q_e \cdot \frac{\tilde{x}_e}{x_e}$. Assuming that this holds for every permutation $\bar{\pi}_v$, it immediately implies the following and proves the claim.
\begin{align*}
\Pr[ \text{ $e$ is the first edge that is realized in $\pi'_v$}] &= \frac{\tilde{x}_e}{x_e} \cdot \Pr[ \text{ $e$ is the first edge that is realized in $\pi_v$}] \\
&= \frac{\tilde{x}_e}{x_e} \cdot x_e = \tilde{x}_e \,.
\end{align*}

We now have to show that for any fixed permutation $\bar{\pi}_v$, we have $q'_e = q_e \cdot \frac{\tilde{x}_e}{x_e}$. If $e \notin \bar{\pi}_v$, then $q_e=0$. In this case $e$ also does not appear in the permutation $\bar{\pi}'_v$, and we have $q'_e=0$ which clearly satisfies our claim. For an edge $e \in \bar{\pi}_v$, the probability that $e$ is the first realized edge is $q_e= p_e \cdot \prod_{e' \in \bar{\pi}_v : \bar{\pi}_v(e') < \bar{\pi}_v(e)} (1-p_{e'})$. This is because $e$ should get realized, however previous edges in the permutation should not get realized.

Now consider the permutation $\bar{\pi}'_v$. In order for an $e$ to be the first edge that gets realized in $\bar{\pi}'_v$, we should satisfy the following conditions.
\begin{itemize}
\item When \DIST reaches $e$, it should add $e$ to the permutation. This happens with probability $\frac{\tilde{x}_e}{x_e}$.
\item When we query $e$ it should gets realized which happens with probability $p_e$.
\item Consider an edge $e'$ such that  $\bar{\pi}_v(e') < \bar{\pi}_v(e)$. Then when \DIST reaches $e'$, either $e'$ is not added to the permutation and \DIST is not ended in Line 7, or either $e'$ is added to the permutation, but it does not get realized when we query it. Note that the probability that $e'$ gets added to the permutation and does not get realized is $(1-p_{e'}) \cdot \frac{\tilde{x}_{e'}}{x_{e'}}$. Also, the probability that $e'$ does not get added to the permutation, and \DIST does not get ended in Line 7 is $(1-p_{e'}) (1- \frac{\tilde{x}_{e'}}{x_{e'}})$. Thus, we satisfy this property for an edge $e'$ with probability
\begin{align*}
(1-p_{e'}) \cdot \frac{\tilde{x}_{e'}}{x_{e'}} + (1-p_{e'}) (1- \frac{\tilde{x}_{e'}}{x_{e'}}) = (1-p_{e'}) \,.
\end{align*}
\end{itemize} 
According to the conditions that we discussed, $q'_e= p_e \cdot \frac{\tilde{x}_e}{x_e} \cdot \prod_{e' \in \bar{\pi}_v : \bar{\pi}_v(e') < \bar{\pi}_v(e)} (1-p_{e'})$. Thus, $q'_e = q_e \cdot \frac{\tilde{x}_e}{x_e}$, and it proves the claim.
\end{proof}

\subsection{Omitted Proofs of Section 5}
\ijerg*
\begin{proof}
Vertex $v$ first draws a random permutation $\pi'_v$ from \DIST. Recall that in \DIST we first draw a random permutation $\pi_v$ from $\mathcal{D}^{\b x}_v$, and then \DIST constructs another random permutation $\pi'_v$ based on it, and returns $\pi'_v$.  Consider any fixed permutation $\bar{\pi}_v$, and let $\bar{\pi}'_v$ be a random permutation that \DIST returns when in Line 2 of \DIST $\pi_v=\bar{\pi}_v$. Let $q$ be the probability that $v$ remains unmatched by \PAM when $e$ is not examined and also $\pi_v = \bar{\pi}_v$. Specifically,
\begin{align*}
q= \Pr [ v \notin \match_1 \ | \ \text{$e$ is not examined} \wedge \pi_v= \bar{\pi}_v ] \,. 
\end{align*}
We then show that $q \ge \frac{1-1/e}{2}$. Since this holds for any fixed permutation $\bar{\pi}_v$, it immediately implies the claim. We first define the following definition.

\begin{definition}
Let $\bar{\pi}_v$ be a fixed permutation, and let $\bar{\pi}'_v$ be a random permutation that \DIST returns when in Line 2 of \DIST $\pi_v=\bar{\pi}_v$. Suppose that vertex $v$ queries the edges in $\bar{\pi}'_v$ in the same order until we find a realized edge. Then, we say that $\bar{\pi}_v$ ``\textit{terminates}''  when it reaches an edge $e$ if either $e$ is queried and got realized or if \DIST gets ended in Line 7 when it reaches $e$.
\end{definition}

We now partition the event that $e$ is not examined and $\pi_v= \bar{\pi}_v$ into following disjoint events, and we show that for each of them the probability that $v$ remains unmatched is at least $\frac{1-1/e}{2}$ which shows that $q\ge \frac{1-1/e}{2}$ and proves the lemma.
\begin{itemize}
\item The first case that $e$ remains unexamined is when $ e \notin \bar{\pi}_v$. In this case if $\bar{\pi}_v$ does not terminate, then $v$ does not find any realized edge. Thus, it remains unmatched and we are done. Otherwise, we can assume that $\bar{\pi}_v$ terminates when it reaches one of the edges, let's say $e'$. When \DIST reaches $e'$, the probability that $e'$ gets added to the permutation and it gets realized is $p_{e'} \cdot \frac{\tilde{x}_{e'}}{x_{e'}}$. It also terminates with the probability $p_{e'}(1-\frac{\tilde{x}_{e'}}{x_{e'}})$ because \DIST ends in line 7. Thus, if we condition that $\bar{\pi}_v$ got terminated when it reached $e'$, the probability that $e'$ is proposed by $v$ is
\begin{align*}
\frac{p_{e'} \cdot \frac{\tilde{x}_{e'}}{x_{e'}}}{p_{e'} \cdot \frac{\tilde{x}_{e'}}{x_{e'}}+p_{e'}(1-\frac{\tilde{x}_{e'}}{x_{e'}})} = \frac{\tilde{x}_{e'}}{x_{e'}} \,.
\end{align*}
When $e'=(v,u')$ gets proposed, as we discussed in the proof of lemma \ref{lem:prm1}, the probability that $e'$ joins the matching is maximized when $u'$ only has one other edge, let's say $e''$, and $x_{e''} = 1- x_{e'}$. In this case $e'$ joins the matching with the probability of at least $ \big (1/2 + 1/2 (1-g(x_{e''},1) \big)$.  Thus, the probability that $v$ gets matched when $\bar{\pi}_v$ gets terminated when it reaches $e'$ is at most
\begin{align*}
&\frac{\tilde{x}_{e'}}{x_{e'}} \big (1/2 + 1/2 (1-g(x_{e''},1) \big) \\
& = \frac{g(x_{e'},1)}{x_{e'}} \big (1/2 + 1/2 (1-g(1-x_{e'},1) \big)\,.
\end{align*}
It is easy to verify that function $g$ always guarantees that the inequality above is at most $1/2+1/(2e)$. Thus, $v$ remains unmatched with the probability of at least $1/2-1/(2e)$. 
\item The second case is when $e \in \bar{\pi}_v$, however $\bar{\pi}_v$ gets terminated before reaching $e$.  It is clear that in this case $e$ does not get examined. Let $T= \{e' \in \bar{\pi}_v : \bar{\pi}_v(e') \le \bar{\pi}_v(e) \}$ be the set of edges that are before $e$ in the permutation $\bar{\pi}_v$. Then, we can say that $\bar{\pi}_v$ got terminated reaching one of the edges in $T$. Similar to the previous case, for every edge $e' \in T$, the probability that $v$ remains unmatched when it terminates  reaching $e'$ is at least $1/2-1/(2e)$. 
\item The third case is when $\bar{\pi}_v$ gets terminated when it reaches $e$. Since we know that $e$ is not examined, then \DIST must be ended in Line 7 when it reaches $e$. In this case $v$ does not propose to any vertex, and always remains unmatched.
\item The last case is when $\bar{\pi}_v$ does not get terminated before or when it reaches $e$, and also it does not add $e$ to the permutation. In this case all of the edges that are before $e$ in the permutation $\bar{\pi}_v$ should not get realized. Also, $v$ does not examine $e$ as it is not added to the permutation. Let $\bar{\pi}''_v$ be the suffix of $\bar{\pi}_v$ formed by the edges that are after $e$ in the permutation $\bar{\pi}_v$. It is then clear that if we condition on this case, proposals of $v$ have the same distribution as when it runs on $\pi_v = \bar{\pi}''_v$. This reduces this case to the first case, and shows $v$ remains unmatched with the probability of $1/2-1/(2e)$.     
\end{itemize}
We showed that for all cases $v$ remains unmatched with the probability  $\frac{1-1/e}{2}$ which proves the lemma. 
\end{proof}

\clmvpnotu*
\begin{proof}
First if we $v'=v$, then $Pr[v' \not \curly u\ |\ \text{$e$ is not examined}]=1$ and the claim clearly holds. Thus, the remaining case is when $v' \neq v$. In this case since the proposal of vertices in $A$ are independent, the event $v' \not \curly u$ is independent from the event that $e$ is not examined. Therefore,
\begin{align*}
\Pr[ v' \not \curly u \ | \ X_{v,u'} \ \wedge \ \text{$e$ is not examined}] = \Pr[ v' \not \curly u \ | \ X_{v,u'}] \,.
\end{align*}
Recall that $X_{v,u'}$ is the event that at least one of the vertices in $\pi_{<v}$ proposes to $u'$. We similarly, use $\bar{X}_{v,u'}$ to denote the compliment of this event where none of vertices in $\pi_{<v}$ proposes to $u'$.  We show that $\Pr[ v' \not \curly u] \le \Pr[ v' \not \curly u \ | \ X_{v,u'}]$ by showing that $\Pr[ v' \not \curly u] \ge \Pr[ v' \not \curly u \ | \ \bar{X}_{v,u'}]$. Note that if $v' \in \pi_{<v}$, the event $\bar{X}_{v,u'}$ from the perspective of $v'$ is equivalent to say that $v'$ has not proposed to $u'$, i.e., $v' \not \curly u'$. On the other hand, if $v' \notin \pi_{<v}$ then the event $\bar{X}_{v,u'}$ does not tell us anything about the proposal of vertex $v'$. Thus, we can calculate the probability $\Pr[ v' \not \curly u \ | \ \bar{X}_{v,u'} \ ]$ as follows. 
\begin{align}
\Pr[ v' \not \curly u \ | \ \bar{X}_{v,u'}] &= \Pr[ v' \not \curly u\ |\ v' \in \pi_{<v} \wedge \bar{X}_{v,u'}] \Pr[v' \in \pi_{<v} \ | \ \bar{X}_{v,u'}] \nonumber \\
&+ \Pr[ v' \not \curly u\ |\ v' \notin \pi_{<v} \wedge \bar{X}_{v,u'}] \Pr[v' \notin \pi_{<v} \ | \ \bar{X}_{v,u'}] \nonumber \\
&= \Pr[ v' \not \curly u\ |\ v' \not \curly u'] \Pr[v' \in \pi_{<v} \ | \ \bar{X}_{v,u'}] \nonumber\\
\label{eq:vnotubarx}
&+ \Pr[ v' \not \curly u] \Pr[v' \notin \pi_{<v} \ | \ \bar{X}_{v,u'}] \,.
\end{align} 
Since vertex $v'$ proposes to at most one of the vertices in $B$, we have
\begin{align*}
\Pr[ v' \not \curly u\ |\ v' \not \curly u'] &= 1- \Pr[ v' \curly u\ |\ v' \not \curly u'] \\&= 1- \frac{\Pr[v' \curly u]}{\sum_{u'' \in B} \Pr [ v' \curly u'' ]- \Pr[v' \curly u']} \\&\ge  1- \Pr[v' \curly u] \\&= \Pr[v' \not \curly u] 
\end{align*}
Using the inequality above and (\ref{eq:vnotubarx}), we have
\begin{align*}
\Pr[ v' \not \curly u \ | \ \bar{X}_{v,u'}] &= \Pr[ v' \not \curly u\ |\ v' \not \curly u'] \Pr[v' \in \pi_{<v} \ | \ \bar{X}_{v,u'}]\\
&+ \Pr[ v' \not \curly u] \Pr[v' \notin \pi_{<v} \ | \ \bar{X}_{v,u'}]  \\
&\ge \Pr[ v' \not \curly u] \Pr[v' \in \pi_{<v} \ | \ \bar{X}_{v,u'}]\\
&+ \Pr[ v' \not \curly u] \Pr[v' \notin \pi_{<v} \ | \ \bar{X}_{v,u'}] \\
& = \Pr[ v' \not \curly u] \,.
\end{align*}
This show that $\Pr[ v' \not \curly u] \ge \Pr[ v' \not \curly u \ | \ \bar{X}_{v,u'}]$. As a direct result, $\Pr[ v' \not \curly u] \le \Pr[ v' \not \curly u \ | \ X_{v,u'}]$ which proves the claim.
\end{proof}

\subsection{Omitted Proofs of Section 6}
\gxderiv*
\begin{proof}
By taking the derivative of $\frac{g(x,\ubar)}{x}$ w.r.t. $x$ we get the following
%\textcolor{red}{derivative is commented out!}
\begin{align*}
\frac{\odv}{\odv x} \frac{g(x,\ubar)}{x} = 
\frac{\odv}{\odv x} \frac{(e^\ubar-1)(\ubar  -x)}{\ubar (e^\ubar- e^x)}  = - \frac{(e^{\ubar}-1)(e^{\ubar}+ e^x (x-1-\ubar))}{\ubar (e^{\ubar}-e^{x})^2} \,.
\end{align*}
We show that the derivative above is negative for $0 \le x < \ubar$ which implies that the function is strictly decreasing. For an $0 \le x < \ubar$ we always have $e^{x}(x-1-\ubar) > -e^{\ubar}$. This implies that $(e^{\ubar}+ e^x (x-1-\ubar)) >0$ for $0 \le x < \ubar$, and the derivative taken above is negative $0 \le x < \ubar$.  
\end{proof}

\minsum*
\begin{proof}
We first show the following claim.
\begin{claim}
Let $a,b \in \mathbb{R}_+^2$ be real numbers such that $a+b \le 1$ and $a \ge b$. Also, let $\delta$ a any non-negative real number, i.e., $\delta \ge 0$. Then,
\begin{align*}
(1-(a+\delta)) (1- (b-\delta)) \le (1-a)(1-b) \,.
\end{align*} 
\end{claim}
\begin{proof}
By expanding the left hand side of the desired inequality we get
\begin{align*}
(1-(a+\delta)) (1- (b-\delta)) &= 1- a - b + (a+\delta) (b-\delta) \\
&= 1- a - b + a \cdot b +\delta( b-a - \delta) \,.
\end{align*}
Since $a \ge b$ and $\delta \ge 0$, we have $b-a-\delta \le 0$. Thus,
\begin{align*}
(1-(a+\delta)) (1- (b-\delta)) &= 1- a - b + a \cdot b +\delta( b-a - \delta) \\
&\le  1- a - b + a \cdot b = (1-a)(1-b) \,,
\end{align*}
which proves the claim.
\end{proof}

Let $a_1, a_2, \cdots, a_k \in [0,c]^k$ be $k$ real numbers in $[0,c]^k$. Suppose that there are two numbers $a_i$ and $a_{i'}$ in this sequence such that $0<a_i \le a_{i'} <c$. Then we create  a new sequence $a'_1, a'_2, \cdots, a'_k$ where $a'_j = a_j$ if $j \notin \{i,i'\}$, $a'_i=\max\{0,a_i+ a_{i'}- c\}$ and $a'_{i'}=\min\{c, a_i+a_{i'} \}$. It is clear that $\sum_{i=1}^k a_i= \sum_{i=1}^k a'_i$. Then by the claim above we have  $\prod_{i=1}^k (1- a_i) \ge \prod_{i=1}^k (1- a'_i)$. Therefore the term $\prod_{i=1}^k (1- a_i)$ is minimized when at most one of the $a_i$'s are not equal to $c$ or $0$. In other words,
\begin{align*}
\prod_{i=1}^k (1- a_i) \ge (1-c)^{q} \cdot (1- (s- c \cdot q)) \,,
\end{align*} 
where $q = \lfloor \frac{s}{c} \rfloor$. We then claim that $1-(s- c \cdot q) \ge (1-c)^{\frac{s}{c}-q}$. First we show how this completes the proof the lemma. Assuming that aforementioned inequality holds, we have 
\begin{align*}
\prod_{i=1}^k (1- a_i) &\ge (1-c)^{q} \cdot (1- (s- c \cdot q)) \\
&\ge (1-c)^{q} \cdot (1-c)^{\frac{s}{c}-q} \\
&=(1-c)^{\frac{s}{c}} \,,
\end{align*}
which is the desired bound in the lemma. Therefore, it is remained to show that $1-(s- c \cdot q) \ge (1-c)^{\frac{s}{c}-q}$. To prove this we use the following inequality: For any $ 0 \le x \le 1$ and $0 \le p \le 1$, we have $(1-x)^p \le 1- p \cdot x$.  Note that $\frac{s}{c}-q = \frac{s}{c}- \lfloor \frac{s}{c} \rfloor  < c \le 1$. Therefore, using the mentioned inequality we can say that $(1-c)^{\frac{s}{c}-q} \le 1- (c \cdot (\frac{s}{c}-q)) = 1-(s- c \cdot q)$ which completes the proof.
\end{proof}

\gmovertau*
\begin{proof}
We show that $\gm(\tau) < \tau$ for $3/4 \le \tau \le 1$ which immediately implies the claim. By the definition of $\gm (.)$ function,  $\gm(\tau)$ is the smallest real number $x$ such that $\frac{g(x',1)}{x'} \le \tau$ for any $x' \ge x$. Let $y=0.74$, we show that $\frac{g(y,1)}{y} \le 3/4$. This and the fact that $\frac{g(x,1)}{x}$ is a decreasing function implies that  $\gm(\tau) \le y < \tau$ for $3/4 \le \tau \le 1$ and proves the claim.

It is now remained to show that $\frac{g(y,1)}{y} \le 3/4$. By substituting function $g$ we get
\begin{align*}
\frac{g(y,1)}{y} &= \frac{(1-1/e)(1-y)}{1-e^{y-1}} \approx 0.718 \le 3/4 & \text{$y$=0.74}
\end{align*}
\end{proof}

\rholessone*
\begin{proof}
Recall that  $\rho(\tau)$ is the largest positive real value $x \in (0,1]$ such that $x \le 1-(1-\frac{\gm(\tau)}{\tau})^{\big(\frac{x}{\gm(\tau)}\big)}$. Thus, to show that $\rho(\tau) <1 $ it is enough to show that $1-(1-\frac{\gm(\tau)}{\tau})^{\big(\frac{1}{\gm(\tau)}\big)} < 1$. For $3/4 \le \tau \le 1$, by Claim \ref{clm:gmovertau} we have $\frac{\gm(\tau)}{\tau} <1$. Also, $\gm (\tau) >0 $ for $\tau <1$. Therefore, $(1-\frac{\gm(\tau)}{\tau})^{\big(\frac{1}{\gm(\tau)}\big)} >0 $, and $1-(1-\frac{\gm(\tau)}{\tau})^{\big(\frac{1}{\gm(\tau)}\big)} < 1$ which proves the claim.
\end{proof}

\rhofinal*
\begin{proof}
First we claim that for $0.31719 \le \gm (\tau) \le 0.3172$ for $\tau=0.8723$. This can be easily verified as $\frac{g(0.3172,1)}{0.3172} \approx 0.872296 < \tau$ and $\frac{g(0.31719,1)}{0.31719} \approx 0.8723003 > \tau$. Recall that  $\rho(\tau)$ is the largest positive real value $x \in (0,1]$ such that $x \le 1-(1-\frac{\gm(\tau)}{\tau})^{\big(\frac{x}{\gm(\tau)}\big)}$. In other words, it is the largest $x$ such that $x - 1 +(1-\frac{\gm(\tau)}{\tau})^{\big(\frac{x}{\gm(\tau)}\big)} \le 0$. Using the lower and upper bounds that we derived for $\gm$, and by setting $x=0.5303$ we get
\begin{align*}
&x - 1 +(1-\frac{\gm(\tau)}{\tau})^{\big(\frac{x}{\gm(\tau)}\big)}\\
 &\ge x-1+ (1-\frac{0.3172}{\tau})^{\big(\frac{x}{0.31719}\big)}\\
&\approx 6.3 \cdot 10^{-7} >0  \,. & \text{By setting $x=0.5303$ and $\tau=0.8723$}
\end{align*}
Therefore, for $x=0.5303$ we have $x - 1 +(1-\frac{\gm(\tau)}{\tau})^{\big(\frac{x}{\gm(\tau)}\big)} >0 $. We claim that for any $x'>x$ we also have $x' - 1 +(1-\frac{\gm(\tau)}{\tau})^{\big(\frac{x'}{\gm(\tau)}\big)} >0$ which immediately implies that $\rho(\tau) \le x =0.5303$ and proves the claim. To prove this we take the derivative of $x - 1 +(1-\frac{\gm(\tau)}{\tau})^{\big(\frac{x}{\gm(\tau)}\big)}$ and show it is always positive for $x \ge 0.5303$. By taking the derivative of the discussed function we get
\begin{align*}
&\frac{\odv}{\odv x} \bigg(x - 1 +(1-\frac{\gm(\tau)}{\tau})^{\big(\frac{x}{\gm(\tau)}\big)} \bigg) \\
& = 1+ \frac{(1-\frac{\gm(\tau)}{\tau})^{\big(\frac{x}{\gm(\tau)}\big)} \cdot \ln (1- \frac{\gm(\tau)}{\tau})}{\gm(\tau)} \\
& \ge 1+ \frac{(1-\frac{0.31719}{\tau})^{\big(\frac{x}{0.3172}\big)} \cdot \ln (1- \frac{0.31719}{\tau})}{0.31719} & \text{Since $0.31719 \le \gm (\tau) \le 0.3172$} \\
& \ge 1+ \frac{(1-\frac{0.31719}{\tau})^{\big(\frac{0.5303}{0.3172}\big)} \cdot \ln (1- \frac{0.31719}{\tau})}{0.31719} & \text{The second term is increasing w.r.t $x$ and $x \ge 0.5303$} \\
&\approx 1- 0.6693 >0 \,. 
\end{align*} 
This completes the proof of the claim.
\end{proof}

\section{The Transformation Function $g(x_e, \ubar)$}\label{sec:transform}
In this section, we discuss how the transformation function $g$ has been derived. We show that function $g$ is the only function that satisfies the following two properties.

\begin{enumerate}
\item Function $g$ should not change the $x_e$ when $x_e$ is arbitrary small, i.e., for every $\sigma \in [0,1]$, we should have $\lim_{x_e \rightarrow 0} g(x_e, \ubar) = x_e$.
\item Consider an edge $e=(v,u)$, and let $e'_1, e'_2, \cdots, e'_k$ be other edges incident to $u$. As we discussed we always add dummy edges to the vertices $B$, such that $x_u = \ubar$ for every $u \in B$ where $x_u \leftarrow \sum_{e \ni u} x_e$. Another crucial property that we want to satisfy is that when $k$ approaches infinity and $x_{e'_i}$ is arbitrary small for each $e'_i$, then \PAM algorithm matches $e$ exactly with probability $\frac{(1-e^{-\ubar})x_e}{\ubar}$. Let $\match_1=$\PAM{$G, \b{x}, \sigma$} and $x_{e'_i} = \frac{\ubar - x_e}{k}$, then we want to satisfy the following property.
\begin{align*}
\lim_{k \rightarrow \infty} \Pr [ e \in \match_1 ] = \frac{(1-e^{-\ubar})x_e}{\ubar} \,.
\end{align*}
\end{enumerate}

Now we show that 
$$g(x, \ubar) = \frac{(e^\ubar-1)(\ubar  -x) x}{\ubar (e^\ubar- e^x)}$$
is the unique function that satisfies these two properties. 

Consider second desired property and an edge $e=(v,u)$. Then \PAM matches $e$ whenever it gets proposed and none of the edges $e'_i$ that are before $e$ in $\pi$ gets proposed. Let $\delta= \ubar - x_e$. Then in the second property for every other edge $e'_i$ incident to $u$ we have $x_{e'_i}= \frac{\delta}{k}$.  We then have,
\begin{align*}
\Pr[e \in \match_1] &= \Pr[e \text{ is proposed}]\sum_{i=0}^k\Pr[\text{There are } i \text{ edges before } e \text{ in  } \pi \wedge \text{none of them are proposed}]
\\&= g(x_e)\times \left(\sum_{i=0}^{k}\frac{1}{k+1} \left(1-g\left(\delta/k\right)\right)^i\right) \,.
\end{align*}

The probability of $e \in \match_1$ when $k \rightarrow \infty$ is then equal to
\begin{align*}
\lim_{k \rightarrow \infty} \Pr[e \in \match_1]= g(x_e, \ubar)\times \lim_{k\rightarrow \infty}\left(\sum_{i=0}^{k}\frac{1}{k+1} \left(1-g\left(\delta/k\right)\right)^i\right) \,.
\end{align*}

Note that when $k$ approaches infinity, $\delta/k$ approaches $0$. Therefore, by the first desired property of function $g$ we should have $\lim_{k \rightarrow \infty} g(\delta/k, \ubar) = \delta/k$. Thus,
\begin{align*}
&\lim_{k \rightarrow \infty} \Pr[e \in \match_1]= g(x_e,\ubar)\times \lim_{k\rightarrow \infty}\left(\sum_{i=0}^{k}\frac{1}{k+1} \left(1-\left(\delta/k\right)\right)^i\right)
\\& = \frac{g(x_e, \ubar)}{k}\times \lim_{k \rightarrow \infty}\left(\sum_{i=0}^{k} \left(1-\left( \delta/k \right)\right)^i\right) 
\\&= g(x_e, \ubar)\times \lim_{k \rightarrow \infty} \left( \sum_{i=0}^{k} \frac{e^{-i\delta/k}}{k} \right)
\\ & =  g(x_e, \ubar)\times \frac{1-e^{-\delta}}{\delta} \,. & \text{By (\ref{eq:integrate})}
\end{align*}

Recall that the second desired property of function $g$ is $\lim_{k \rightarrow \infty} \Pr [ e \in \match_1 ] = \frac{(1-e^{-\ubar})x_e}{\ubar}$. According to the probability above, in order to satisfy this property we should have
\begin{align*}
&g(x_e, \ubar)= \frac{\delta (1- e^{-\sigma}) x_e}{\sigma (1-e^{-\delta})} \\
&=  \frac{(\ubar - x_e) (1- e^{-\sigma}) x_e}{\sigma (1-e^{x_e - \ubar})} \\
&=  \frac{(e^\ubar-1)(\ubar  -x_e) x_e}{\ubar (e^\ubar- e^{x_e})} \,.
\end{align*}

\section{The Price of Information (PoI)}
\label{sec:poi}
In this section, we discuss how our techniques can be used to beat $(1-1/e)$ approximation for the matching constraints in the price of information (PoI) model. In this problem, we are given a set of edges $E$ and a feasibility system $\mathcal{F} \subseteq 2^E$. Here we assume that the feasibility set $\mathcal{F}$ is the set of all valid matchings in $E$. For each edge $e \in E$, we are also given a search cost $c_e \in \mathbb{R}^+$ and a distribution $F_e$ for the value of the edge. An algorithm can query an edge to realize its valuation $v_e$ from the distribution $F_e$, the algorithm then has to pay the cost $c_e$.  Let $S$ be the set of edges queried by the algorithm then the utility of an algorithm is equal to  $\max_{\match \in S: \match \in F} \sum_{e \in \match} v_e- \sum_{e \in S} c_e$. The goal is to maximize the expected utility.  

In comparison to the query commit model, the POI problem has search costs, however it does not require commitment to the queried edges. Gamlath et al. ~\cite{DBLP:conf/soda/GamlathKS19} provides a method to reduce the matchings in PoI to the stochastic matching problem with query commitment. Their method finds a threshold $\chi_e$ for each edge, and then changes the weight of every edge to $\kappa_e=\min\{v_e, \chi_e\}$. 

Without loss of generality, assume that the $\kappa_e$ is discrete. Let $K_e$ be the set of all possible values of $\kappa_e$, their approach replaces every edge with an edge-value pair $(e,v)$ where $v \in K_e$, and they derive the LP below.

\setcounter{equation}{0}
\vspace{3mm}
%\newlength{\LPlhbox}
\settowidth{\LPlhbox}{(LP-Match)}%
\noindent%
\hspace*{\fill}%
\begin{minipage}{\linewidth-2cm}
 \begin{align}
\nonumber \max_{(e,v) \in E_{\text{all}}} x_{e,v} \cdot v	\,& \\
\nonumber \textrm{s.t.} \qquad & \sum_{(e,v) \in F} x_{e,v} \le f(F)& \forall u \in A \cup B,  \forall F \subseteq E_u  \\
\nonumber &x_{e,v} \ge 0& \forall (e,v) \in E_{\text{all}}
\end{align}
\end{minipage}
Where $f(F)$ is the probability that $\kappa_e = v$ for at least one edge-pair value $(e,v) \in F$. 
They show that there is a distribution $\mathcal{D}$ for every vertex in $B$ such that it picks an edge $e$ that has the value $v$ with probability $x^*_{e,v}$.

Consider the analogous of our algorithm using their reduction. It can be verified that our approach beats the $(1-1/e)$ approximation by a constant if $|K_e|=O(1)$, i.e., when each discrete distribution takes constant number of possible values. Apply our transformation function to get modified solution $\tilde{x}_{e,v}= g(x_{e,v}, \sigma)$, and let $\tau$ be a constant very close to $1$. Consider running our two-round algorithm for this modified solution. It beats the $(1-1/e)$ approximation if at least a constant fraction of the optimal solution is available at the second round. As discussed the edges that do not proceed to the second round are the edges with almost $\tilde{x}_{e,v}/p_{e,v}=1$. Suppose that almost of the contribution of the optimal solution is among these edges. Now suppose that we run our one-round algorithm for only these edges. The algorithm beats $(1-1/e)$ approximation if for every vertex in $B$ the summation of its incident edges $\sum_{(e,v)} x_{e,v}$ is less than one. For the stochastic matching problem this was true because of the constraint \eqref{cons:437489} of LP. The same fact holds if $|K_e|=O(1)$ beating the $(1-1/e)$ approximation. However, the same claim does not hold when $K_e$ is very large. For example suppose that we have a graph with a single edge $e$, and the $\kappa_{e}$ is uniform over $[0,c]$ where $c$ is a very large constant. Then, in the discretization of the distribution we get a lot of small discrete parts. Then, in the solution of LP $x_{e,v}$ would be small values bounded by $p_{e,v}$. Since $x_{e,v}$ is small for these edges, our transformation function almost does not change them and we have $\tilde{x}_{e,v} \approx x_{e,v}$. So the value of $\tilde{x}_{e,v}/p_{e,v}$ is larger than $\tau$ for these edges. However, the summation $\sum_{(e,v)} x_{e,v}$ can be also $1$, since in the distribution $\kappa_e$ it is guaranteed that one of the edge-value pairs gets realized.  
\vspace{5mm}

\end{document}